\documentclass[journal]{IEEEtran}

\usepackage{cite}
\usepackage{amsmath,amssymb,amsfonts}
\usepackage{graphicx}
\usepackage{textcomp}
\usepackage{xcolor,soul}
\usepackage{epsfig}
\usepackage{multirow}
\usepackage{balance}
\usepackage{float}
\usepackage{algpseudocode, algorithm, tabularx}
\usepackage{algorithmicx}
\usepackage[caption=false]{subfig}
\usepackage{varwidth}
\usepackage{amssymb}
\usepackage{amsthm}
\usepackage{mathtools}
\usepackage{url}

\usepackage[breaklinks]{hyperref}

\DeclareMathOperator*{\minimize}{minimize}
\DeclareMathOperator*{\maximize}{maximize}

\newtheorem{lemma}{Lemma}

\makeatletter
\newcommand{\multiline}[1]{%
	\begin{tabularx}{\dimexpr\linewidth-\ALG@thistlm}[t]{@{}X@{}}
		#1
	\end{tabularx}
}
\makeatother

\begin{document}

\title{Wirelessly Powered Federated Learning Networks: \\Joint Power Transfer, Data Sensing, Model Training, and Resource Allocation}
\author{Mai~Le, Dinh~Thai~Hoang,~\IEEEmembership{Senior Member,~IEEE},
Diep~N.~Nguyen,~\IEEEmembership{Senior Member,~IEEE}, Won-Joo~Hwang,~\IEEEmembership{Senior Member,~IEEE}, and Quoc-Viet~Pham,~\IEEEmembership{Member,~IEEE}
\vspace{-0.5cm}
\thanks{Mai Le is with the Department of Information Convergence Engineering, Pusan National University, Busan 46241, Korea, and also with the School of Computer Science and Statistics, Trinity College Dublin, The University of Dublin, D02 PN40, Ireland (e-mail: maile2108@gmail.com).}
\thanks{Dinh~Thai~Hoang, and Diep~N.~Nguyen are with the School of Electrical and Data Engineering, University of Technology Sydney, Sydney, NSW 2007, Australia (e-mail: \{hoang.dinh, diep.nguyen\}@uts.edu.au).}
\thanks{Won-Joo Hwang (corresponding author) is with the Department of Biomedical Convergence Engineering, Pusan National University, Yangsan 50612, Korea (e-mail: wjhwang@pusan.ac.kr).}
\thanks{Quoc-Viet Pham is with the School of Computer Science and Statistics, Trinity College Dublin, The University of Dublin, D02 PN40, Ireland (e-mail: viet.pham@tcd.ie).}
\thanks{The first and last authors contribute equally to this work.}
}

\maketitle

\begin{abstract}
Federated learning (FL) has found many successes in wireless communications; however, the implementation of FL has been hindered by the energy limitation of mobile devices (MDs) and the availability of training data at MDs. Wireless power transfer (WPT) and mobile crowdsensing (MCS) are promising technologies that can be leveraged to power energy-limited MDs and acquire data for learning tasks. However, how to integrate WPT and MCS towards sustainable FL solutions is a research topic entirely missing from the open literature. This work for the first time investigates a resource allocation problem in collaborative sensing-assisted sustainable FL (S2FL) networks with the goal of minimizing the total completion time. In particular, we investigate a practical harvesting-sensing-training-transmitting protocol in which energy-limited mobile devices first harvest energy from RF signals, use it to gain a reward for user participation, sense the training data from the environment, train the local models at mobile devices, and transmit the model updates to the edge server. The total completion time minimization problem of jointly optimizing power transfer, transmit power allocation, data sensing, bandwidth allocation, local model training, and data transmission is complicated due to the non-convex objective function, highly non-convex constraints, and strongly coupled variables. 
We propose a computationally-efficient path-following algorithm to obtain the optimal solution via the decomposition technique. In particular, inner convex approximations are developed for the resource allocation subproblem, and the subproblems are performed alternatively in an iterative fashion. Simulation results are provided to evaluate the effectiveness of the proposed S2FL algorithm in reducing the completion time up to 21.45\% in comparison with other benchmark schemes. Further, we investigate an extension of our work from frequency division multiple access (FDMA) to non-orthogonal multiple access (NOMA) and show that NOMA can speed up the total completion time 8.36\% on average of the considered FL system. 
\end{abstract}

\begin{IEEEkeywords}
Artificial Intelligence, Federated Learning, Mobile Crowdsensing, Mobile Edge Computing, Resource Allocation, Wireless Power Transfer.
\end{IEEEkeywords}

\IEEEpeerreviewmaketitle
\section{Introduction}
\label{Sec:Introduction}
\IEEEPARstart{A}{rtificial} intelligence (AI) has found many successes in various engineering disciplines, of which future wireless communications are a notable area \cite{wang2023road, de20226g}. The increasing usage of AI techniques has been seen as a key to solving various challenging problems in 5G networks that are difficult to be handled by conventional optimization approaches, as well as meeting technical requirements in future 6G networks and Internet-of-Things (IoT). It has been shown that AI techniques, such as deep learning and reinforcement learning, can be used to solve various challenging problems in wireless communications that are hardly handled by conventional optimization techniques. However, conventional AI techniques require mobile devices (MDs) to share their data with a learning server for centralized storage and learning, which, however, results in significant issues of communication bottleneck, storage availability, and especially data privacy. Federated learning (FL) has been considered a promising AI technique to overcome these issues \cite{kairouz2021advances}. Unlike conventional AI techniques, FL aims to build a collaborative learning model by allowing MDs to share their model updates instead of raw data, thus enhancing data privacy \cite{qin2021federated}. As such, there is an increasing interest in making use of FL in wireless networks, such as vehicular networks \cite{xu2023secure}, cybersecurity \cite{alazab2022federated}, and wireless resource allocation \cite{yang2021energyFL}. 

However, vanilla FL frameworks are severely limited by the energy limitation of MDs (e.g., sensors and wearable devices) and the data available for local training. To overcome these hurdles, it is extremely promising to integrate FL with key-enabling technologies of future wireless networks like wireless power transfer (WPT) and simultaneous wireless information and power transfer (SWIPT) to enhance the energy capacities of energy-limited MDs and thus increase the FL performance \cite{wu2022swipt}. For example, \cite{pham2021uav, pham2022energy} proposed to dispatch an unmanned aerial vehicle (UAV) as the learning server in FL and as the aerial energy source in WPT. These studies show that aerial FL helps to improve the network performance compared with the benchmarks, especially when the placement of the learning server is predefined. WPT-aided FL was also investigated in recent works \cite{wu2022simultaneous, li2022dynamic, wu2022non}. In \cite{wu2022simultaneous}, A WPT and non-orthogonal multiple access (NOMA)-aided FL system was considered, where each learning round consists of four phases, including WPT, local training, non-orthogonal transmission, and model aggregation. Following that, a layered algorithm was proposed to optimize the transmit power of the edge server, time allocation, processing rates, and local training accuracy so as to minimize the energy-latency trade-off objective. Motivated by \cite{wu2022simultaneous}, the work in \cite{li2022dynamic} studied a long-term optimization problem of user scheduling, power allocation, and power splitting factor, which is then solved by an actor-critic reinforcement learning method. However, unlike WPT in \cite{pham2022energy, wu2022simultaneous}, SWIPT was adopted in \cite{li2022dynamic} to allow the edge server to broadcast the global model and wirelessly power MDs simultaneously. Different from \cite{li2022dynamic}, which fixes the time of the four phases, the work in \cite{wu2022non} proposed to optimize the time of the SWIPT phase and NOMA transmission phase. 
However, these prior studies were based on the assumption that the data is available for local training at MDs. With the growing sensing capabilities of MDs, it is apparent that MDs can sense data from the environment, which becomes available for training local models. Nevertheless, the sensing data for FL demands a novel design of properly allocating time for different learning phases, along with communication and computation resources.

Mobile crowdsensing (MCS) has been a key technology for the future Internet that allows MDs to perform data sensing, computing, and communication tasks in order to complete more complex tasks \cite{capponi2019survey}. With multi-access edge computing (MEC) deployment, computing resources are made available at the network edge, and the sensing data can be offloaded to a server for remote processing. Conventionally, the sensing data is collected by centralized AI algorithms to build intelligent MCS systems, thus raising the data privacy issue. In this regard, FL not only enables intelligent collaborative sensing systems but also helps to preserve the data privacy of MDs. Moreover, resource allocation optimization has shown a critical role in improving the performance of MCS systems \cite{zhang2022fedsky, zhao2022crowdfl, wang2021learning, gao2021federated, li2022joint, mao2022intelligent}. For example, the work in \cite{zhang2022fedsky} proposed a novel aggregation method in FL-enabled MCS systems to address the system's significant issues of appropriate selection of heterogeneous MDs and lack of privacy-preserving methods. 
In \cite{zhao2022crowdfl}, a hybrid incentive method based on a reverse Vickrey auction and a posted pricing mechanism was proposed to promote the participation of MDs as workers in MCS systems. In \cite{wang2021learning}, a blockchain-enabled FL-MCS network was investigated, in which blockchain and differential privacy are employed to decentralize and secure the transmission of the global model and local model updates. More recently, studies on joint data sensing, computing, and communication are investigated in \cite{li2022joint, mao2022intelligent}, which confirm the superiority of the joint framework over baselines with disjoint considerations. 
It is worth noting that prior studies \cite{capponi2019survey, zhang2022fedsky, zhao2022crowdfl, wang2021learning, gao2021federated, li2022joint, mao2022intelligent} mainly focus on leveraging FL to improve data privacy or edge computing resources to process heterogeneous sensing data in MCS and ignore the usage of collaborative sensing for FL via a joint power transfer, data sensing, model training, and communication framework. It is a relatively unexplored research question about how to leverage WPT and MCS for FL and how to optimize the resource allocation problem in such FL networks with the objective of minimizing the total completion time.

On the other hand, designing efficient transmission to share model updates from MDs to the edge server is a crucial issue in FL. In this regard, prior studies have adopted OMA, NOMA, or both \cite{yang2021energyFL, pham2022energy, liang2022data, chen2022irs, mao2022intelligent}. From the optimization perspective, NOMA consistently achieves a better sum rate than conventional OMA techniques \cite{chen2017optimization}; however, that observation may not hold for MEC and FL systems with different performance metrics. For example, the work in \cite{yang2021energyFL} considered optimizing the total energy consumption and the total completion time in FDMA-based FL networks and showed that FDMA is mostly superior to time-division multiple access (TDMA), which may outperform FDMA in terms of the computation energy consumption in case of the high power budget of MDs. The works in \cite{liang2022data, chen2022irs} investigated the optimization problem of computation offloading in MEC systems with two kinds of offloading schemes: NOMA and TDMA. These studies show that TDMA can be superior to NOMA under specific scenarios of identical sensing rates in \cite{liang2022data} and flexibly optimized beamforming design in \cite{chen2022irs}. The work in \cite{mao2022intelligent} studied the completion time minimization problem in reconfigurable intelligent surface (RIS) aided FL networks with two multiple access schemes, including frequency division multiple access (FDMA) and NOMA. Simulation results in \cite{mao2022intelligent} exhibit that the NOMA-FL system experiences a shorter completion time than the FDMA-FL system, while RIS further helps to reduce the completion time compared with ones without RIS. 

\begin{figure*}[t]
	\centering
	\includegraphics[width=0.7635\linewidth]{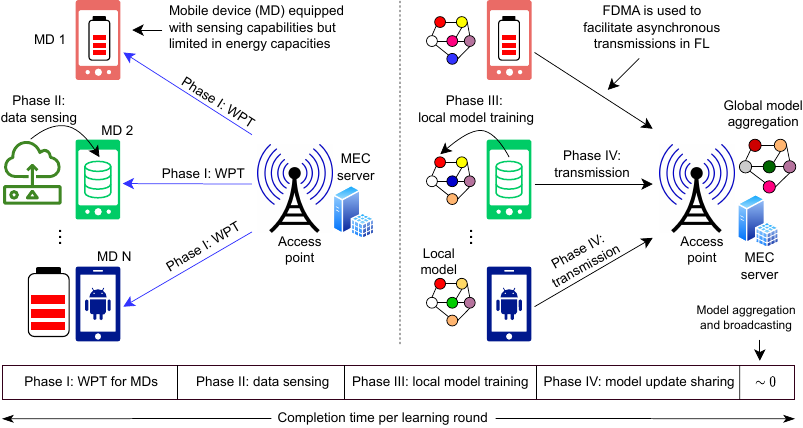}
	\caption{Illustration of the system model in a communication round that consists of $N$ MDs and one edge server collocated with the access point.} 
	\label{Fig:System_model}
\end{figure*}

As aforementioned, an efficient FL implementation is constrained by MDs' limited battery capacity and training data availability. WPT and MCS are two promising technologies to overcome these deficiencies so as to enable intelligent IoT services and applications in future wireless systems. However, to the best of our knowledge, the study of integrating WPT and MCS towards sustainable FL solutions is a research topic entirely missing from the open literature. \textit{Different from prior studies that assume the availability of training data, the sensing process is integrated into FL networks to exploit the growing sensing capabilities of MDs in this paper}. To this end, we aim to investigate a sensing-assisted sustainable FL (called S2FL) scenario with the implementation of the four-phase of ``harvesting-sensing-training-transmitting''. Specifically, in our investigated S2FL system, an operator-deployed edge server collocated with an access point employs WPT to wirelessly power MDs in the first phase. After receiving harvested energy as a reward, the MDs are incentivized to contribute to the learning process by using the rest of harvested energy to sense data, train local models, and share local model updates. Further, FDMA-based transmission and an extension to the non-orthogonal multiple access (NOMA) counterpart are considered for the communication between MDs and the edge server. Under this design, it is crucial to properly allocate the time for WPT, sensing, local training, and transmission because of a tight coupling among the learning phases. It is also important to take into account the processing rates of MDs (i.e., central processing unit (CPU) frequencies), their transmit power, power transfer, and bandwidth partitioning. Following the aforementioned considerations, the completion time minimization problem is formulated in this paper, and an efficient path-following algorithm is proposed to obtain the optimal solution. In summary, notable contributions and features offered by this work are summarized as follows:
\begin{itemize}
    \item \textbf{Consideration of S2FL}: Two main challenges of efficient FL implementation are the energy limitation and availability of training data at MDs. To settle these issues, we investigate a practical scenario in FL networks when the MDs first harvest energy from the energy source via RF signals, then leverage it to gain a reward, sense data from the environment, next train local models using the sensing data, and use residual energy to share local model updates with the learning server. Considering these aspects, the four-phase ``harvesting-sensing-training-transmitting'' protocol is proposed to design the S2FL framework.
    
    \item \textbf{Problem Formulation of S2FL Networks}: With the aim of minimizing the total completion time of the considered S2FL network, we formulate an optimization problem with a set of constraints on energy, data sensing, model training, communication, and computing resources. However, obtaining the globally optimal solution for the problem is computationally challenging because of its non-convexity and a tight coupling among optimization variables. As such, we devise a computationally-efficient path-following algorithm to effectively deal with the non-convexity issue and obtain the optimal solution. In particular, we propose to decompose the original problem into subproblems of resource allocation and local model accuracy and develop inner convex approximations to generate the convex programs of these two problems. The algorithm is then designed in the alternative fashion of updating the local model accuracy and solving the convex approximated resource allocation problem. 
    
    \item \textbf{OMA-enabled and NOMA-enabled in S2FL}: Inspired by the superiority of NOMA over OMA, we extend the S2FL system to employ NOMA as the multiple access scheme for data transmission from MDs to the edge server so as to reduce the total completion time further. Unlike the OMA case, MDs are allowed to use the entire bandwidth resource for communication with the edge server, and the bandwidth allocation is thus not needed to be considered, but the optimization problem is still highly non-convex. By leveraging the decomposition technique and inner convex approximations as in the OMA case, we then propose a computationally-efficient path-following algorithm for the S2FL problem in the NOMA case. 
    
    \item \textbf{Extensive Performance Evaluation}: Various simulation results are provided to show the effectiveness of our proposed S2FL framework. Moreover, the proposed algorithm is compared with a set of benchmark schemes under various network settings. It is shown that the proposed algorithm offers great performance improvements over the benchmark schemes. In particular, we provide sub-optimal schemes for the S2FL problem, including fixed training data, given model accuracy, proportional power transfer, and equal bandwidth allocation. Through a set of simulation results under various network settings, we show that the proposed S2FL can significantly minimize the total completion time compared with the benchmark schemes. We also show that the NOMA-enabled S2FL system exhibits a lower completion time than the FDMA-enabled S2FL counterpart.
\end{itemize}

The rest of this paper is organized as follows. In Section~\ref{Sec:SystemModel}, the system model and the optimization problem are presented. Then, we introduce our proposed path-following algorithm in Section~\ref{Sec:Solution}, and extend the algorithm to the NOMA case in Section~\ref{Sec:Extension_NOMA}. Next, various simulation results are presented in Section~\ref{Sec:Evaluation} to verify the effectiveness of the proposed algorithm. Finally, we conclude this work and highlight several interesting directions for future work in Section~\ref{Sec:Conclusion}. 

\section{Preliminaries and Problem Formulation}
\label{Sec:SystemModel}
In this section, we present the models of energy harvesting, data sensing, local computing, and communication and then formulate a non-convex optimization problem of jointly optimizing the power allocation, computing rate, bandwidth allocation, time allocation, and model accuracy. 

\subsection{Network Model}
\label{SubSec:SystemModel_NM}

As illustrated in Fig.~\ref{Fig:System_model}, we consider a collaborative sensing network with an edge computing server and $N$ mobile devices (MDs). The operator-deployed MEC server is collocated with the access point at the network edge, and the set of MDs/IoT devices is denoted as $\mathcal{N} = \{1, \dots, N\}$. 
Due to the limited battery capacity, MDs may not have sufficient energy to complete the FL process. To overcome this issue, radio frequency (RF)-based energy harvesting is a promising approach that allows MDs to harvest external energy from RF signals transmitted by the access point. Therefore, each MD can participate in and perform the FL process based on the energy converted from the received RF signals.
In addition, unlike vanilla FL, which assumes the availability of data for local model training at MDs, this work considers a more practical scenario, where MDs with sensing capabilities can sense and then collect the data from the surrounding environment before training their local models. For example, biosensors can be utilized to sense healthcare data, and cameras embedded in MDs can be used to collect image data.

Under the network setup above, we propose a new harvesting-sensing-training-transmit protocol. More specifically, the entire FL process is composed of four main phases with adjustable time durations, as illustrated in Fig.~\ref{Fig:System_model}, including energy harvesting, data sensing, local model training, and model update sharing. 
\begin{itemize}
    \item \textit{Energy harvesting}: At the beginning of each learning round, MDs receive radio frequency (RF) signals transmitted by the access point for a duration of $\tau^{h}$ and convert them into energy for further usage. Furthermore, the energy beamforming technique is employed to improve energy transfer efficiency. Therefore, the access point can simultaneously beam energy to the energy-limited MDs so that they can secure sufficient energy for sensing and learning purposes in the sequel.
    
    \item \textit{Data sensing}: After the WPT phase, MDs can leverage the harvested energy (as a reward) and a part to sense data from the environment for a period of $\tau^{s}$. By the end of this phase, each MD can collect a certain amount of data that will not be shared with the edge server but used to train its local model. 
    
    \item \textit{Local model training}: In the third phase, the sensing data is used to train local models for $\tau^{l}$. Note that the sensing data is not shared with the edge server so that privacy can be enhanced compared with centralized AI techniques. 
    
    \item \textit{Transmission of model updates}: After completing the local training phase, MDs share their local model updates with the edge server to create a new global model. This phase lasts for a period of $\tau^{t}$.
\end{itemize}
It is noteworthy that there are other phases of model aggregation and model broadcasting; however, these phases can be ignored due to the powerful computing resources and large power budget of the MEC server, i.e., the period of these two phases can be approximated as zero, or the phases can be performed by leveraging the energy harvesting phase. Moreover, for the sake of the joint power transfer, data sensing, model training, and transmission framework proposed in this work, we consider a synchronous design with a certain degree of synchronization among MDs. An asynchronous design may allow MDs to have more flexibility in performing the four phases in S2FL, and this design is left for future work. 

\begin{figure}[t]
	\centering
	\includegraphics[width=0.925\linewidth]{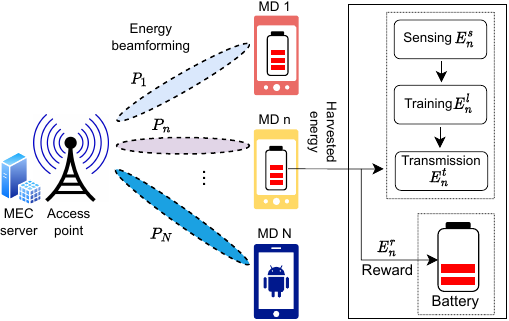}
	\caption{Our considered wirelessly powered S2FL system.}
	\label{Fig:WPT_reward}
\end{figure}

\subsection{Energy Harvesting}
Due to the battery limitation, MDs need more energy to complete the learning process, and RF energy harvesting is a promising solution. During the first phase of the harvesting-sensing-training-transmit protocol, MDs harvest energy from RF signals emitted by the access point. Specifically, the access point simultaneously transfers RF energy to MDs by employing the energy beamforming technique to create $N$ radio beams (see Fig.~\ref{Fig:WPT_reward}). Denote by $P_{n}$ the transmit power of the beam pointing at MD $n$, and by $P_{0}$ the fixed power budget of the access point. Thus, the following power constraint is imposed
\begin{equation}
    \sum\nolimits_{n=1}^{N}P_{n} \leq P_{0}.
\end{equation}
Without loss of generality, we consider a linear energy harvesting model, a widely adopted model for WPT in the open literature \cite{pham2021uav, pham2022energy, li2019wirelessly}. Thus, the amount of harvested energy of MD $n$ can be calculated as follows\footnote{The harvested energy can be characterized by a more general nonlinear model. In such a case, the harvested energy increases when the input power increases and then becomes saturated when the input power is sufficiently large. Based on innerly convex approximations developed in \cite{pham2022energy}, solving the S2FL problem considered in this work with a nonlinear energy harvesting model is relatively  straightforward. It is also promising to extend our current work to consider multi-source WPT to improve the FL performance.}:
\begin{equation}
    E_{n}^{h} = \varphi\tau^{h} P_{n} h_{n},
\end{equation}
where $\varphi$ ($0 < \varphi < 1$) denotes the energy conversion efficiency coefficient and $h_{n}$ denotes the channel gain between the MEC server and MD $n$. Then, MD $n$ stores a part of $E_{n}^{h}$ in its battery as the reward and uses the rest to perform data sensing, local training, and transmission. Denote by $E_{n}^{r}$, $E_{n}^{s}$, $E_{n}^{l}$, $E_{n}^{t}$ the respective energy consumption of the reward, sensing, local training, and transmission phases. Based on the discussion above, the following constraint should hold for each MD in each learning round
\begin{equation}
    E_{n}^{r} + E_{n}^{s} + E_{n}^{l} + E_{n}^{t} \leq E_{n}^{h}.
\end{equation}
This constraint ensures that each MD $n$ cannot use more than its harvested energy for further usage in reward, data sensing, local training, and model update sharing. Here, the energy reward is supposed to be proportional to the amount of data $d_{n}$ to be sensed, i.e., $E_{n}^{r} = d_{n}q_{n}^{r}$, where $q_{n}^{r}$ is a reward coefficient. 

\subsection{Data Sensing}
In the second phase, i.e., after harnessing enough energy, MDs continue learning by sensing the environment's data. Denote by $r_{n}$ the sensing rate of MD $n$, the sensing data (in bits) can be given as $d_{n} = r_{n} \tau^{s}$\footnote{Since the data sensed by MDs may not reflect the global data distribution well, it is worthwhile to investigate a sampling-driven control for FL that considers the generalization gap. This is a promising research direction that we are pursuing in another work.}. Let $E_{n}^{s}$ be the energy consumed for data sensing at MD $n$, which is proportional to the size of the sensing data. Thus, $E_{n}^{s}$ is calculated as follows: 
\begin{equation} \label{Eq:sensingdata_energy}
    E_{n}^{s} = q_{n}^{s}d_{n} = q_{n}^{s} r_{n} \tau^{s},
\end{equation}
where $q_{n}^{s}$ denotes the energy consumption per sensing data unit. Note that the sensing data is used for local training at MDs and is not shared with the server for centralized storage and training. In this way, data collected externally is only exploited for local training purposes while its privacy is enhanced via the use of FL. Furthermore, it is noted that the design problem of optimizing the data sensing time is equivalent to one that optimizes the sensing data of all MDs \cite{li2019wirelessly}. Thus, this work can be readily extended to an interesting scenario where the sensing data sizes of MDs are considered. 

\textbf{Local Computing}: 
After data sensing, the collected data is used for local training in the third phase. Denote by $C_{n}$ and $f_{n}$ the computation workload (in CPU cycles per bit) and computing capability (in CPU cycles per second) of MD $n$. Let $\tau_{n}^{l}$ be the local computing time of MD $n$ that can be expressed as 
\(\tau_{n}^{l} = \eta_{\text{loc}}{C_{n}d_{n}}/{f_{n}},\)
where $\eta_{\text{loc}}$ is the number of local learning iterations required by MD $n$ to achieve a desirable accuracy $\eta$ of the local model. Note that $\eta_{\text{loc}}$ can be specified based on the loss function of MDs and learning parameters, which will be detailed in the next part. 
We also note that the time of the third phase is given as $\tau^{l} = \max_{n \in \mathcal{N}} \{\tau_{n}^{l}\}$.
Following a practical model that was widely adopted in the literature \cite{pham2019coalitional, yang2021energyFL, pham2022energy}, the energy consumption in the local computing phase is given as follows: 
\begin{equation} \label{Eq:lc_energy}
    E_{n}^{l} = \eta_{\text{loc}}\zeta_{n}C_{n}d_{n}f_{n}^{2},
\end{equation}
where $\zeta_{n}$ is an energy coefficient depending on the CPU architecture of MD $n$, and $\zeta_{n}f_{n}^{2}$ expresses the energy consumption per CPU cycle. It is noticed that unlike prior studies on resource allocation for FL, where the training data size is given in advance, the training data in this work is dependent on the period of the second phase, which, in turn, affects the local training time and local energy consumption. Therefore, the completion time and energy consumption of MD $n$ can be adjusted by optimizing the computing rate, the sensing time, and the local model accuracy.

From the theoretical viewpoint, large training (i.e., sensing) data is beneficial for the learning process in AI/FL, especially in the later training rounds \cite{devarakonda2017adabatch, shi2022talk}. Nevertheless, too large training data at an FL iteration (i.e., batch) may not achieve as good performance as one can do with small training data and also greatly increases the local computing time and energy consumption. It is therefore necessary to decide on an appropriate sensing size, which in turn determines the sensing time in the considered S2FL network. In this work, it is assumed that each MD needs to sense an amount of $D_{0}$ data bits for local training at each learning round - in other words, the following constraint holds
\begin{equation} \label{Eq:trainingdata}
    d_{n} \geq D_{0} , \forall n \in \mathcal{N}.
\end{equation}
Note that an adaptive control of the sensing data sizes can further increase the learning performance, which is promising research but out of the scope of this work.

\subsection{Communication Model}
After the second phase of local computing, the model updates are transmitted to the edge server for model aggregation by employing the FDMA scheme. Denote by $B$ and $b_{n}$ the system bandwidth and the bandwidth assigned to MD $n$, respectively, i.e., $b_{n} \geq 0$ and $\sum\nolimits_{n \in \mathcal{N}}b_{n} \leq B$. The transmission rate of MD $n$ is given as
\begin{equation} \label{Eq:rate_FDMA}
    R_{n} = b_{n} \log_{2}\left( 1 + \frac{h_{n}p_{n}}{n_{0}b_{n}} \right),
\end{equation}
where $p_{n}$ is the transmit power of MD $n$ and $n_{0}$ is the noise power density. Let $s$ be the size of the local models, and it is supposed to be the same for all MDs \cite{yang2021energyFL, pham2022energy}. Then, the transmission rate in \eqref{Eq:rate_FDMA} is required to satisfy the following constraint in order to ensure successful transmission of the model updates to the edge server within the fourth phase: 
\begin{equation}
    R_{n} \tau^{c} \geq s, \forall n \in \mathcal{N}.
\end{equation}
The corresponding energy consumption to transmit the model update of size $s$ to the MEC server is $E_{n}^{c} = p_{n} \tau^{c}$. 

\subsection{Problem Formulation}
\label{SubSec:SystemModel_PF}
In the previous part, models necessary to formulate the optimization problem are presented. It is now assumed that the loss function of MDs is $L$-Lipschitz continuous and $\gamma$-strongly convex. As such, the minimum number of local rounds to achieve the accuracy of $\eta$ is given as \cite{yang2021energyFL, pham2022energy} 
\begin{equation} \label{Eq:localrounds}
    \eta_{\text{loc}} = \frac{2}{(2 - L \delta)\delta \gamma } \log_{2} \left(\frac{1}{\eta }\right),
\end{equation}
where the step $\delta$ is required to be less than $2 / L$, i.e., $\delta < 2 / L$. Once the local training process completes, MDs share the local model updates with the edge server for model aggregation. Then, the global model can reach an accuracy $\epsilon_{0}$ after a certain number of global communication rounds as 
\begin{equation} \label{Eq:globalrounds}
    \eta_{\text{glo}} = \frac{\text{2}L^{2}}{\gamma ^{2}\xi } \log \left(\frac{1}{\epsilon _{0}}\right) \frac{1}{1-\eta },
\end{equation}
where $\xi$ is a hyper-learning parameter that satisfies $\xi < \gamma / L$, and $\log(\cdot)$ denotes the natural logarithm. The proofs for the expressions \eqref{Eq:localrounds} and \eqref{Eq:globalrounds} in detail can be found in \cite{yang2021energyFL}; the way is to apply the first-order Taylor approximation and find the relative difference of the training problem's objective value in two consecutive iterations. It is observed from \eqref{Eq:globalrounds} that a high local model accuracy, i.e., a small value of $\eta$, can greatly boost the learning process, i.e., minifying the number of global rounds. However, a small value of $\eta$ may raise the issue of sizeable local computing resources and significantly expand the number of local rounds and the local training time, as expressed in \eqref{Eq:localrounds}. Therefore, it becomes important to optimize the local model accuracy to minimize the total completion time in S2FL networks while satisfying a couple of design constraints on transmit power, time allocation, computing rate, and bandwidth allocation.

We denote the optimization vectors of power transfer, transmit power allocation, bandwidth allocation, and time allocation as $\boldsymbol{P} \triangleq \{P_{n}\}$, $\boldsymbol{p} \triangleq \{p_{n}\}$, $\boldsymbol{f} \triangleq \{f_{n}\}$, $\boldsymbol{b} \triangleq \{b_{n}\}$, and $\boldsymbol{\tau} \triangleq \{\tau^{h}, \tau^{s}, \tau^{l}, \tau^{c}\}$, respectively. Here, the specific design of the optimization problem is to jointly optimize $(\boldsymbol{P}, \boldsymbol{p}, \boldsymbol{f}, \boldsymbol{b}, \boldsymbol{\tau}, \eta)$ to minimize the total completion time of the learning process. Here, the completion time for a given communication round is $\tau^{h} + \tau^{s} + \tau^{l} + \tau^{c}$, and the total completion time is given as 
\begin{equation} \label{Eq:totaltime}
    O(\boldsymbol{\tau},\eta) \triangleq \eta_{\text{glo}} (  \tau^{h} + \tau^{s} + \tau^{l} + \tau^{c}),
\end{equation}
which is considered as the objective function in the S2FL problem. 
Mathematically, we can formulate the optimization problem in the considered S2FL system as follows: 
\begin{subequations}
    \label{OptPrb}
	\begin{align} 
	& \underset{\boldsymbol{P}, \boldsymbol{p}, \boldsymbol{f}, \boldsymbol{b}, \boldsymbol{\tau}, \eta}{\minimize}
	& & \eta_{\text{glo}} (  \tau^{h} + \tau^{s} + \tau^{l} + \tau^{c}), \label{OptPrb:Obj}\\
	& \text{subject to} 
	& & E_{n}^{r} + E_{n}^{s} + E_{n}^{l} + E_{n}^{c} \leq E_{n}^{h}, \forall n \in \mathcal{N}, \label{OptPrb:energy}\\
	&&& \tau^{l} \geq \eta_{\text{loc}}{C_{n}d_{n}}/{f_{n}}, \forall n \in \mathcal{N}, \label{OptPrb:localtime}\\
	&&& \tau^{c}R_{n} \geq s, \forall n \in \mathcal{N}, \label{OptPrb:datasize}\\
	&&& \sum\nolimits_{n=1}^{N}P_{n}\tau^{h} \leq E_{\max},  \label{OptPrb:maxenergy}\\
	&&& r_{n}\tau^{s} \geq D_{0}, \forall n \in \mathcal{N}, \label{OptPrb:sensingdatasize}\\
	&&& 0 \leq p_{n} \leq p_{n}^{\max}, \forall n \in \mathcal{N}, \label{OptPrb:power}\\
	&&& \sum\nolimits_{n=1}^{N}P_{n} \leq P_{0}, P_{n} \geq 0, \forall n \in \mathcal{N}, \label{OptPrb:power_BS}\\
	&&& \sum\nolimits_{n \in \mathcal{N}}b_{n} \leq B, \; b_{n} \geq 0, \forall n \in \mathcal{N}, \label{OptPrb:bAllocation}\\
	&&& f_{n}^{\min} \leq f_{n} \leq f_{n}^{\max}, \forall n \in \mathcal{N}, \label{OptPrb:CPUfre}\\
	&&& 0 < \eta \leq 1. \label{OptPrb:accuracy}
	\end{align}
\end{subequations}
The constraint \eqref{OptPrb:maxenergy} guarantees that the consumption cost of the energy source in each learning round does not exceed an upper bound $E_{\max}$. 
The feasible ranges of transmit power, bandwidth, and CPU frequency are constrained in \eqref{OptPrb:power}, \eqref{OptPrb:bAllocation}, and \eqref{OptPrb:CPUfre}, respectively, where $B$ denotes the system bandwidth, $p_{n}^{\max}$ is the maximum transmit power of MD $n$, and $f_{n}^{\min}$ ($f_{n}^{\max}$) is the minimum (maximum) computing capability of MD $n$. The last constraint \eqref{OptPrb:accuracy} expresses that the local model accuracy is limited within the feasible range $(0,1]$. As aforementioned, the local model accuracy is optimized to boost the learning process, but meanwhile, it satisfies the constraints on energy consumption and the local learning process. 

The problem~\eqref{OptPrb} is non-convex due to the non-convex objective function and the non-convex constraints. For example, there is a coupling between the local model accuracy and the time allocation variables, which makes the objective function non-convex. 
For the constraints, the transmission time, the bandwidth allocation, and the transmit power are coupled in \eqref{OptPrb:datasize}, and thus \eqref{OptPrb:datasize} is a non-convex constraint. 
As such, no globally optimal solution, in general, can solve the problem efficiently. To this end, we develop a computationally-efficient path-following algorithm for the S2FL problem in \eqref{OptPrb}. Beforehand, it should be noted that the constraint on the sensing data size \eqref{OptPrb:sensingdatasize} is active at the MD with the slowest sensing rate; because otherwise, one can always reduce the sensing time until the constraint holds for all MDs and thus reduces the total completion time and learning overhead. 

\section{Joint Power Transfer, Data Sensing, Model Training, and Transmission}
\label{Sec:Solution}
As aforementioned, there are no efficient algorithms to obtain globally optimal solutions in general. In this section, we present an efficient algorithm to solve the highly complex non-convex problem in \eqref{OptPrb}. 
Particularly, we develop inner convex approximations of the non-convex constraints, based on which we devise a path-following algorithm for the S2FL optimization problem. 
We first denote $a = -{\text{2}L^{2}} \log \left({\epsilon _{0}}\right) / [{\gamma ^{2}\xi}]$ and $\nu = {2} / [{(2 - L \delta)\delta \gamma } \log{2}]$ to simplify  the expressions of learning rounds, and rewrite the problem as follows: 
\begin{subequations}
    \label{OptPrb_rw}
	\begin{align} 
        & \underset{\boldsymbol{P}, \boldsymbol{p}, \boldsymbol{f}, \boldsymbol{b}, \boldsymbol{\tau}, \eta}{\minimize}
        & & \frac{a}{1 - \eta} \left( \tau^{h} + \tau^{s} + \tau^{l} + \tau^{c} \right), \label{OptPrb_rw:Obj}\\
        & \text{subject to} 
        & & q_{n}^{r}r_{n}\tau^{s} + q_{n}^{s}r_{n}\tau^{s} + \nu\log\left(\frac{1}{\eta}\right) \zeta_{n}C_{n}r_{n}\tau^{s}f_{n}^{2} \notag\\
        &&& \qquad + p_{n} \tau^{c} \leq \varphi\tau^{h} P_{n} h_{n}, \forall n \in \mathcal{N}, \label{OptPrb_rw:energy}\\
        &&& \tau^{l} \geq \nu\log\left(\frac{1}{\eta}\right) \frac{C_{n} r_{n}\tau^{s}}{f_{n}}, \forall n \in \mathcal{N}, \label{OptPrb_rw:localtime}\\
        &&& \tau^{c}b_{n} \log_{2}\left( 1 + \frac{h_{n}p_{n}}{n_{0}b_{n}} \right) \geq s, \forall n \in \mathcal{N}, \label{OptPrb_rw:datasize}\\ 
        &&& \eqref{OptPrb:maxenergy}, \eqref{OptPrb:sensingdatasize}, \eqref{OptPrb:power}, \eqref{OptPrb:power_BS}, \eqref{OptPrb:bAllocation}, \eqref{OptPrb:CPUfre}, \eqref{OptPrb:accuracy}. \notag
	\end{align}
\end{subequations}
To solve this non-convex problem, we develop inner convex approximations to devise an efficient solution approach. In particular, we propose to decompose the problem into subproblems and propose to solve them alternatively in an iterative manner. Specifically, the entire set of optimization variables $(\boldsymbol{P}, \boldsymbol{p}, \boldsymbol{f}, \boldsymbol{b}, \boldsymbol{\tau}, \eta)$ is decomposed into two subsets: 1) the first one includes the local model accuracy $\eta$; and 2) the second one of resource allocation includes the time allocation $\boldsymbol{\tau}$ and the computing rates (i.e., CPU frequencies) $\boldsymbol{f}$ of MDs, the bandwidth allocation $\boldsymbol{b}$, the transmit power allocation of MDs $\boldsymbol{p}$, and the transmit power of the access point $\boldsymbol{P}$. The fact is that the local model accuracy remains unchanged for a much longer duration compared with that of resource allocation \cite{vu2020cell}. Exploiting this timescale property, the proposed algorithms can be implemented in a two-timescale manner when the local model accuracy is updated at a long timescale, and the latter is updated at a short timescale.

Let $\kappa$ be the iteration index, and $x^{(\kappa)}$ be the value of the variable $x$ at step $\kappa$. 
In the following, we demonstrate the development of our proposed iterative algorithm with inner convex approximations of the non-convex constraints.

\subsection{Long-timescale Problem to Optimize Model Accuracy}
Given the values $(\boldsymbol{P}, \boldsymbol{\tau}, \boldsymbol{f}, \boldsymbol{b}, \boldsymbol{p})$ from the previous step $(\kappa-1)$, we find the local model accuracy at step $\kappa$.  
The model training problem to find the local model accuracy can be expressed as follows: 
\begin{subequations}
    \label{OptPrb_eta}
	\begin{align} 
	& \underset{\eta, \tau^{l}}{\minimize}
	& & \frac{a}{1 - \eta} \left( \iota^{(\kappa-1)} + \tau^{l} \right), \label{OptPrb_eta:Obj}\\
	& \text{subject to} 
	& & \nu\log\left({1}/{\eta}\right) \zeta_{n}C_{n}r_{n}\tau^{s}f_{n}^{2} \leq \psi_{n}, \forall n \in \mathcal{N}, \label{OptPrb_eta:energy}\\
	&&& \tau^{l} \geq \nu\log\left(\frac{1}{\eta}\right) \frac{C_{n} r_{n}\tau^{s}}{f_{n}}, \forall n \in \mathcal{N}, \label{OptPrb_eta:localtime}\\
	&&& 0 < \eta \leq 1, \label{OptPrb_eta:accuracy}
	\end{align}
\end{subequations}
where $\iota^{(\kappa-1)} = \tau^{h, (\kappa-1)} + \tau^{s, (\kappa-1)} + \tau^{c, (\kappa-1)}$, and $\psi_{n} \triangleq \varphi\tau^{h} P_{n}^{(\kappa-1)} h_{n} - p_{n} \tau^{c} - q_{n}^{s}r_{n}\tau^{s} - q_{n}^{r}r_{n}\tau^{s}$. Note that the constraint \eqref{OptPrb_eta:localtime} and the local computing time $\tau^{l}$ are considered in this problem as well as the resource allocation problem of the second block. It is because the variable $\tau^{l}$ does not only affect the local model accuracy but also decides the completion time in a given learning round. 
Since $\log(1/x)$ is a convex function over $\mathbb{R}^{+}$ \cite{boyd2004convex}, the constraints \eqref{OptPrb_eta:energy} and \eqref{OptPrb_eta:localtime} are convex. Accordingly, all the constraints of the problem~\eqref{OptPrb_eta} are convex; however, the objective function is not convex due to its fractional nature. 

For approximation purposes, the objective function $f(\eta, \tau^{l})$ can be rewritten as follows:
\begin{equation} \label{Eq:OptPrb_eta_eq_Obj_rw}
    f(\eta, \tau^{l}) = \frac{a \iota^{(\kappa-1)}}{1 - \eta} + \frac{a \tau^{l}}{1 - \eta}. 
\end{equation}
While the first term in \eqref{Eq:OptPrb_eta_eq_Obj_rw} is convex with respect to $\eta$, the second term is not convex. 
Fortunately, we leverage the following inequality to convexify the second term
\begin{align} \label{Ineq:trickyineq}
    \frac{a \tau^{l}}{1 - \eta}
    & = \frac{a\tau^{l,(\kappa-1)}}{4(1 - \eta^{(\kappa-1)})} \notag\\
    & \quad \times \Biggl( \left( \frac{\tau^{l}}{\tau^{l,(\kappa-1)}} + \frac{1 - \eta^{(\kappa-1)}}{1 - \eta} \right)^{2} \notag\\
    & \quad\qquad - \left( \frac{\tau^{l}}{\tau^{l,(\kappa-1)}} - \frac{1 - \eta^{(\kappa-1)}}{1 - \eta} \right)^{2} \Biggl) \notag\\
    & \leq \frac{a\tau^{l,(\kappa-1)}}{4(1 - \eta^{(\kappa-1)})} \left( \frac{\tau^{l}}{\tau^{l,(\kappa-1)}} + \frac{1 - \eta^{(\kappa-1)}}{1 - \eta} \right)^{2}  \notag\\
    & \triangleq j^{(\kappa)}(\tau^{l}, \eta).
\end{align}
By a number of simple calculations, one can prove that the Hessian of $j^{(\kappa)}(\tau^{l}, \eta)$ is positive semidefinite, $j^{(\kappa)}(\tau^{l}, \eta)$ is therefore a convex function \cite{boyd2004convex}. 
As a result, $f(\eta, \tau^{l})$ can be innerly approximated as follows: 
\begin{align}
    f(\eta, \tau^{l}) 
    \leq \frac{a \iota^{(\kappa-1)}}{1 - \eta} + j^{(\kappa)}(\tau^{l}, \eta) \triangleq f^{(\kappa)}(\eta, \tau^{l}),
\end{align}
which is the sum of two convex functions, thus being a concave function.
Therefore, we solve the following optimization problem to update the local model accuracy $\eta^{(\kappa)}$ at step $\kappa$
\begin{subequations}
    \label{OptPrb_eta_cvx}
	\begin{align} 
	& \underset{\eta, \tau^{l}}{\maximize}
	& & f^{(\kappa)}(\eta, \tau^{l}), \label{OptPrb_eta_cvx:Obj}\\
	& \text{subject to} 
	& & \eqref{OptPrb_eta:energy}, \eqref{OptPrb_eta:localtime}, \eqref{OptPrb_eta:accuracy}. \notag
	\end{align}
\end{subequations}
This problem is convex; therefore, we can obtain the optimal solution using the CVX package \cite{Grant2013CVX}. 

\subsection{Short-timescale Problem to Optimize Resource Allocation}
Given the local model accuracy, we find the optimal values of the time allocation $\boldsymbol{\tau}$, and the computing rates $\boldsymbol{f}$, bandwidth allocation $\boldsymbol{b}$, transmit power $\boldsymbol{p}$ of MDs, and transmit power of the access point $\boldsymbol{P}$ at step $\kappa$. Given $\eta^{(\kappa)}$ at step $\kappa$, the numbers of local learning rounds $\eta_{\text{loc}}$ and global communication rounds $\eta_{\text{glo}}$ can be calculated in advance according to \eqref{Eq:localrounds} and \eqref{Eq:globalrounds}, respectively. On this point, the objective can be simplified as one of minimizing the completion time of a learning round. Specifically, we rewrite the resource allocation problem of $(\boldsymbol{P}, \boldsymbol{\tau}, \boldsymbol{f}, \boldsymbol{b}, \boldsymbol{p})$ as follows: 
\begin{subequations}
    \label{OptPrb_noneta}
	\begin{align} 
        & \underset{\boldsymbol{P}, \boldsymbol{p}, \boldsymbol{f}, \boldsymbol{b}, \boldsymbol{\tau}}{\minimize}
        & & \eta_{\text{glo}}^{(\kappa)} \left( \tau^{h} + \tau^{s} + \tau^{l} + \tau^{c} \right), \label{OptPrb_noneta:Obj}\\
	& \text{subject to} 
        & & q_{n}^{r}r_{n}\tau^{s} + q_{n}^{s}r_{n}\tau^{s} + \eta_{\text{loc}}^{(\kappa)} \zeta_{n}C_{n}r_{n}\tau^{s}f_{n}^{2} + p_{n} \tau^{c} \notag\\
        &&& \qquad \leq \varphi\tau^{h} P_{n} h_{n}, \forall n \in \mathcal{N}, \label{OptPrb_noneta:energy}\\
        &&& \tau^{l} \geq \eta_{\text{loc}}^{(\kappa)}C_{n} r_{n} \frac{\tau^{s}}{f_{n}}, \forall n \in \mathcal{N}, \label{OptPrb_noneta:localtime}\\
        &&& \tau^{c}b_{n} \log\left( 1 + \frac{h_{n}p_{n}}{n_{0}b_{n}} \right) \geq \tilde{s}, \forall n \in \mathcal{N}, \label{OptPrb_noneta:datasize}\\ 
        &&& \eqref{OptPrb:maxenergy}, \eqref{OptPrb:sensingdatasize}, \eqref{OptPrb:power}, \eqref{OptPrb:power_BS}, \eqref{OptPrb:bAllocation}, \eqref{OptPrb:CPUfre}, \notag
	\end{align}
\end{subequations}
where $\tilde{s} = s \log2$, $\eta_{\text{glo}}^{(\kappa)} = a / (1 - \eta^{(\kappa)})$, $\eta_{\text{loc}}^{(\kappa)} = \nu \log(1/ \eta^{(\kappa)})$, and $\boldsymbol{\tau}$ is redefined as $\boldsymbol{\tau} \triangleq \{\tau^{h}, \tau^{s}, \tau^{l}, \tau^{c}\}$. It can be observed that the local computing time variable $\tau^l$ is considered in both subproblems \eqref{OptPrb_eta} and \eqref{OptPrb_noneta} since it does affect not only the local model accuracy but also other resource allocation parameters. 
Although the objective function and the last three constraints are convex, the problem in~\eqref{OptPrb_noneta} is still non-convex due to the non-convexity of the first three constraints. Beforehand, we would note that the constraints \eqref{OptPrb:maxenergy}, \eqref{OptPrb:sensingdatasize}, \eqref{OptPrb:power}, \eqref{OptPrb:power_BS}, \eqref{OptPrb:bAllocation}, and \eqref{OptPrb:CPUfre} are convex for given $\eta$. It is now tasked with innerly approximating the non-convex constraints \eqref{OptPrb_noneta:energy}, \eqref{OptPrb_noneta:localtime}, and \eqref{OptPrb_noneta:datasize} to build a convex program of \eqref{OptPrb_noneta}.

We introduce the following lemma to build the inner convex approximations of the non-convex constraints in \eqref{OptPrb_noneta}. 
\begin{lemma}
Over the feasible domain $\{x > 0, y > 0\}$, $\sqrt{xy}$ is a concave function. Considering that $\bar{x} > 0$ and $\bar{y} > 0$, the following inequality holds:
\begin{equation} \label{Ineq:sqrt_xy}
    \sqrt{xy} \leq \frac{1}{2}\left( \frac{\sqrt{\bar{y}}}{\sqrt{\bar{x}}}x + \frac{\sqrt{\bar{x}}}{\sqrt{\bar{y}}}y \right).
\end{equation}
\end{lemma}
\begin{proof}
Following the first-order condition of a concave function \cite{boyd2004convex}, we have 
\begin{align}
    f(x,y) 
    & \triangleq \sqrt{xy} \notag\\ 
    & \leq f(\bar{x},\bar{y}) + \langle\nabla f(\bar{x},\bar{y}), (x,y) - (\bar{x},\bar{y})\rangle \notag\\
    & = \frac{1}{2}\left( \frac{\sqrt{\bar{y}}}{\sqrt{\bar{x}}}x + \frac{\sqrt{\bar{x}}}{\sqrt{\bar{y}}}y \right),
\end{align}
which is approximated around the feasible point $\{\bar{x} > 0, \bar{y} > 0\}$. Thus, we obtain the inequality in~\eqref{Ineq:sqrt_xy}. The proof ends.
\end{proof}
Based on~\eqref{Ineq:sqrt_xy} and by setting $t = \sqrt{x}$, $z^{2} = \sqrt{y}$, $\bar{t} = \sqrt{\bar{x}}$, and $\bar{z}^{2} = \sqrt{\bar{y}}$, we can also achieve the following inequality
\begin{equation} \label{Ineq:xy2}
    {tz^{2}} \leq \frac{1}{2}\left( \frac{\bar{z}^{2}}{{\bar{t}}}t^{2} + \frac{{\bar{t}}}{\bar{z}^{2}}z^{4} \right).
\end{equation}
Similarly, based on~\eqref{Ineq:sqrt_xy} and by setting $t = \sqrt{x}$, $z = \sqrt{y}$, $\bar{t} = \sqrt{\bar{x}}$, and $\bar{z} = \sqrt{\bar{y}}$, we can achieve the following inequality
\begin{equation} \label{Ineq:xy}
    {tz} \leq \frac{1}{2}\left( \frac{{\bar{z}}}{{\bar{t}}}t^{2} + \frac{{\bar{t}}}{{\bar{z}}}z^{2} \right).
\end{equation}

\textbf{Convexity of} \eqref{OptPrb_noneta:energy}: 
We follow the inequality in~\eqref{Ineq:xy} with $t = \tau^{s}$, $z = f_{n}$, $\bar{t} = \tau^{s,(\kappa-1)}$, and $\bar{z} = f_{n}^{(\kappa-1)}$ to generate the upper bound of the second term of \eqref{OptPrb_noneta:energy}.
Moreover, by following the inequality in~\eqref{Ineq:xy} with $t = p_{n}$, $z = \tau^{c}$, $\bar{t} = p_{n}^{(\kappa-1)}$, and $\bar{z} = \tau^{c,(\kappa-1)}$, we can generate the upper bound of the third term. As the result, the non-convex constraint in \eqref{OptPrb_noneta:energy} can be innerly approximated as follows:
\begin{align} \label{OptPrb_noneta:energy_app}
    & q_{n}^{s}r_{n}\tau^{s} + q^{(\kappa)}(\tau^{s}, f_{n}) + w^{(\kappa)}(p_{n}, \tau^{c}) \notag\\
    & \qquad \leq \varphi\tau^{h} P_{n} h_{n}, \forall n \in \mathcal{N}, 
\end{align}
where we define 
\begin{align}
    & q^{(\kappa)}(\tau^{s}, f_{n}) = \frac{\eta_{\text{loc}}^{(\kappa)} \zeta_{n}C_{n}r_{n}}{2} \notag\\
    & \qquad\qquad\quad \times \left( \frac{(f_{n}^{(\kappa-1)})^{2}}{\tau^{s,(\kappa-1)}}(\tau^{s})^{2} + \frac{\tau^{s,(\kappa-1)}}{(f_{n}^{(\kappa-1)})^{2}}(f_{n})^{4} \right), \\
    & w^{(\kappa)}(p_{n}, \tau^{c}) = \frac{1}{2} \left( \frac{\tau^{c,(\kappa-1)}}{p_{n}^{(\kappa-1)}}(p_{n})^{2} + \frac{p_{n}^{(\kappa-1)}}{\tau^{c,(\kappa-1)}}(\tau^{c})^{2} \right).
\end{align}
However, \eqref{OptPrb_noneta:energy_app} is still non-convex due to the non-convex part on the right-hand side. We introduce the following lemma to approximate the non-convex part of \eqref{OptPrb_noneta:energy_app}.
\begin{lemma}
For all $t > 0$, $z > 0$, $\bar{t} > 0$, $\bar{z} > 0$, the function $f(t,z) = tz$ can be innerly approximated as follows:
\begin{align} \label{Ineq:tz}
    f(t,z) \geq (2\bar{z}\sqrt{\bar{t}})t^{1/2} - (\bar{t} \bar{z}^{2})z^{-1}.
\end{align}
\end{lemma}
\begin{proof}
The fact that quadratic-over-linear function $f(x,y) = x^2/y$ is convex with $y > 0$ \cite{boyd2004convex}. Thus, its lower-bound concave approximation around the point $(\bar{x},\bar{y})$ is 
\begin{align}
    f(x,y) = \frac{x^2}{y} 
    & \geq f(\bar{x},\bar{y}) + \langle\nabla f(\bar{x},\bar{y}), (x,y) - (\bar{x},\bar{y})\rangle \notag\\
    & = \frac{2\bar{x}}{\bar{y}}x - \frac{\bar{x}^{2}}{\bar{y}^{2}}y.
\end{align}
As a result, by setting $x = \sqrt{t}$, $y = 1/z$, $\bar{x} = \sqrt{\bar{t}}$, and $\bar{y} = 1/\bar{z}$, we can deduce the inequality in \eqref{Ineq:tz}. The proof ends.
\end{proof}

Following \eqref{Ineq:tz} with with $t = \tau^{h}$, $z = P_{n}$, $\bar{t} = \tau^{h,(\kappa-1)}$, and $\bar{z} = P_{n}^{(\kappa-1)}$, we have the following approximation 
\begin{align}
    \varphi\tau^{h} P_{n} h_{n} 
    & \geq \varphi h_{n} \left( (2P_{n}^{(\kappa-1)}\sqrt{\tau^{h,(\kappa-1)}})(\tau^{h})^{1/2} \right.\notag\\
    & \left. \qquad\; - (\tau^{h,(\kappa-1)} (P_{n}^{(\kappa-1)})^{2})(P_{n})^{-1} \right) \notag\\
    & \triangleq v^{(\kappa)}(\tau^{h}, P_{n}),
\end{align}
which is a concave function of $(\tau^{h}, P_{n})$.
Therefore, we can innerly approximate the non-convex constraint in \eqref{OptPrb_noneta:energy} as
\begin{align} \label{OptPrb_noneta:energy_app1}
    & q_{n}^{r}r_{n}\tau^{s} +  q_{n}^{s}r_{n}\tau^{s} + q^{(\kappa)}(\tau^{s}, f_{n}) + w^{(\kappa)}(p_{n}, \tau^{c}) \notag\\
    & \qquad \leq v^{(\kappa)}(\tau^{h}, P_{n}), \forall n \in \mathcal{N}, 
\end{align}
which is now a convex constraint.

\textbf{Convexity of} \eqref{OptPrb_noneta:localtime}: 
Similar to \eqref{Ineq:trickyineq}, we solve the non-convexity issue of the second constraint \eqref{OptPrb_noneta:localtime} via the following tricky approximation:
\begin{align}
    \chi_{n}^{(\kappa)}\frac{\tau^{s}}{f_{n}} 
    & = \frac{\chi_{n}^{(\kappa)}}{4}\frac{\tau^{s,(\tau)}}{f_{n}^{(\kappa)}} \notag\\
    & \quad \times \left( \left( \frac{\tau^{s}}{\tau^{s,(\tau)}} + \frac{f_{n}^{(\kappa)}}{f_{n}} \right)^{2} - \left( \frac{\tau^{s}}{\tau^{s,(\tau)}} - \frac{f_{n}^{(\kappa)}}{f_{n}} \right)^{2} \right) \notag\\
    & \leq \frac{\chi_{n}^{(\kappa)}}{4}\frac{\tau^{s,(\tau)}}{f_{n}^{(\kappa)}} \left( \frac{\tau^{s}}{\tau^{s,(\tau)}} + \frac{f_{n}^{(\kappa)}}{f_{n}} \right)^{2} \notag\\
    & \triangleq g^{(\kappa)}(\tau^{s}, f_{n}).
\end{align}
where $\chi_{n}^{(\kappa)} = \eta_{\text{loc}}^{(\kappa)}C_{n} r_{n}$.
Since $g^{(\kappa)}(\tau^{s}, f_{n})$ is a convex function \cite{boyd2004convex}, the non-convex constraint in \eqref{OptPrb_noneta:localtime} can be innerly approximated as follows: 
\begin{equation}
    \tau^{l} \geq g^{(\kappa)}(\tau^{s}, f_{n}), \forall n \in \mathcal{N}.
\end{equation}

\textbf{Convexity of} \eqref{OptPrb_noneta:datasize}: To address the non-convex constraint in  \eqref{OptPrb_noneta:datasize}, we introduce the following lemma.
\begin{lemma}
Over the feasible domain that $x > 0$, $y > 0$, and $z > 0$, the function $f(x,y,z)$ is convex. Considering that $\bar{x} > 0$, $\bar{y} > 0$, and $\bar{z} > 0$, the following inequality holds
\begin{align} \label{Ineq:logxyz}
    z\log\left( 1 + \frac{1}{xy} \right) 
    & \geq 2\bar{z}\log\left(1 + \frac{1}{\bar{x}\bar{y}}\right) \notag\\
    & \;\;\; + \frac{\bar{z}}{1 + \bar{x}\bar{y}}\left(2 - \frac{x}{\bar{x}} - \frac{y}{\bar{y}}\right) \notag\\
    & \;\;\; - \frac{\bar{z}^{2}\log(1 + 1/\bar{x}\bar{y})}{z}.
\end{align}
\end{lemma}
\begin{proof}
Over the domain $x > 0$, $y > 0$, and $t > 0$, the function $f(x,y,t) = t^{-1}\log(1+1/xy)$ is shown to be convex \cite{nasir2019uav}. Thus, by applying the first-order condition of the convex function $f(x,y,t)$, we have 
\begin{align}
    f(x,y,t) 
    & = \frac{1}{t} \log\left(1+\frac{1}{xy}\right) \notag\\
    & \geq f(\bar{x},\bar{y},\bar{t}) + \langle\nabla f(\bar{x},\bar{y},\bar{t}), (x,y,t) - (\bar{x},\bar{y},\bar{t})\rangle \notag\\
    & = \frac{2\log\left(1 + {1}/{\bar{x}\bar{y}}\right)}{\bar{t}}
    + \frac{1}{\bar{t}(1 + \bar{x}\bar{y})}\left(2 - \frac{x}{\bar{x}} - \frac{y}{\bar{y}}\right) \notag\\
    & \quad - \frac{t\log(1 + 1/\bar{x}\bar{y})}{\bar{t}^{2}}.
\end{align}
Setting $t = 1/z$ and $\bar{t} = 1/\bar{z}$ helps to obtain the inequality in~\eqref{Ineq:logxyz}. The proof ends.
\end{proof}
We now leverage the inequality in~\eqref{Ineq:logxyz} to innerly approximate \eqref{OptPrb_noneta:datasize}, which can be equivalently represented as
\begin{equation} \label{OptPrb_noneta_rw:datasize}
    b_{n} \log\left( 1 + \frac{h_{n}p_{n}}{n_{0}b_{n}} \right) \geq \frac{\tilde{s}}{\tau^{c}}.  
\end{equation}
The left-hand side of \eqref{OptPrb_noneta_rw:datasize} is convex, while the right-hand side is still non-convex. 
By applying \eqref{Ineq:logxyz} for the left-hand side with $z = b_{n}$, $x = n_{0}/p_{n}h_{n}$, $y = b_{n}$, $\bar{z} = b_{n}^{(\kappa-1)}$, $\bar{x} = n_{0}/p_{n}^{(\kappa-1)}h_{n}$, and $\bar{y} = b_{n}^{(\kappa-1)}$, we have the following approximation
\begin{align} \label{Ineq:Rate}
    & b_{n}\log\left(1 + \frac{p_{n}h_{n}}{b_{n}n_{0}}\right) \notag\\
    & \geq \lambda_{k} + \mu_{k}\left( 2 - \frac{p_{n}^{(\kappa-1)}}{p_{n}} - \frac{b_{n}}{b_{n}^{(\kappa-1)}} \right) - \frac{\upsilon_{k}}{b_{n}} \notag\\
    & \triangleq R_{n}^{(\kappa)}(b_{n},p_{n}),
\end{align}
where the coefficients are given as follows: 
\begin{align}
    & \lambda_{k} = 2b_{n}^{(\kappa-1)} \log\left(1 + \frac{p_{n}^{(\kappa-1)}h_{n}}{b_{n}^{(\kappa-1)}n_{0}}\right), \\
    & \mu_{k} = b_{n}^{(\kappa-1)}\left(1 + \frac{b_{n}^{(\kappa-1)}n_{0}}{p_{n}^{(\kappa-1)}h_{n}}\right)^{-1}, \\
    & \upsilon_{k} = (b_{n}^{(\kappa-1)})^{2}\log\left(1 + \frac{p_{n}^{(\kappa-1)}h_{n}}{b_{n}^{(\kappa-1)}n_{0}}\right).
\end{align}
It is observed that $R_{n}^{(\kappa)}(b_{n},p_{n})$ is a concave function, and thus the non-convex constraint in~\eqref{OptPrb_noneta:datasize} is innerly approximated as
\begin{equation}
    R_{n}^{(\kappa)}(b_{n},p_{n}) \geq \frac{\tilde{s}}{\tau^{c}}, \forall n \in \mathcal{N},
\end{equation}
which is a convex constraint. 

In brief, we solve the following optimization problem at step $\kappa$ to obtain $(\boldsymbol{P}, \boldsymbol{\tau}, \boldsymbol{f}, \boldsymbol{b}, \boldsymbol{p})$ for a given model accuracy.
\begin{subequations}
    \label{OptPrb_noneta_cvx}
	\begin{align} 
        & \underset{\boldsymbol{P}, \boldsymbol{p}, \boldsymbol{f}, \boldsymbol{b}, \boldsymbol{\tau}}{\minimize}
        & & \eta_{\text{glo}}^{(\kappa)} \left( \tau^{h} + \tau^{s} + \tau^{l} + \tau^{c} \right), \label{OptPrb_noneta_cvx:Obj}\\
	& \text{subject to} 
        & & q_{n}^{r}r_{n}\tau^{s} + q_{n}^{s}r_{n}\tau^{s} + q^{(\kappa)}(\tau^{s}, f_{n}) + w^{(\kappa)}(p_{n}, \tau^{c}) \notag\\
        &&& \qquad \leq v^{(\kappa)}(\tau^{h}, P_{n}), \forall n \in \mathcal{N}, \label{OptPrb_noneta_cvx:energy}\\   
        &&& \tau^{l} \geq g^{(\kappa)}(\tau^{s}, f_{n}), \forall n \in \mathcal{N}, \label{OptPrb_noneta_cvx:localtime}\\
        &&& R_{n}^{(\kappa)}(b_{n},p_{n}) \geq \frac{\tilde{s}}{\tau^{c}}, \forall n \in \mathcal{N}, \label{OptPrb_noneta_cvx:datasize}\\
        &&& \eqref{OptPrb:maxenergy}, \eqref{OptPrb:sensingdatasize}, \eqref{OptPrb:power}, \eqref{OptPrb:power_BS}, \eqref{OptPrb:bAllocation}, \eqref{OptPrb:CPUfre}. \notag
	\end{align}
\end{subequations}
This optimization problem is convex; therefore, we can employ the CVX package to solve the problem efficiently \cite{Grant2013CVX}. Moreover, solving the convex problem in \eqref{OptPrb_noneta_cvx} helps to generate the feasible solution $(\boldsymbol{P}^{(\kappa)}, \boldsymbol{\tau}^{(\kappa)}, \boldsymbol{f}^{(\kappa)}, \boldsymbol{b}^{(\kappa)}, \boldsymbol{p}^{(\kappa)})$, which is then used to build the approximations and find the feasible solution in the next iteration $(\kappa+1)$.

\begin{figure}[t]
	\centering
	\includegraphics[width=0.910\linewidth]{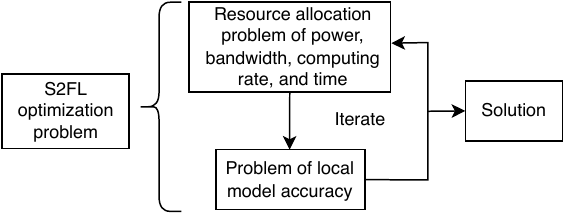}
	\caption{Solution approach.}
	\label{Fig:Solution_approach}
\end{figure}

\subsection{Proposed Iterative Algorithm}
The pseudocode of the proposed algorithm is presented in Algorithm~\ref{Alg:FDMA}, and the flow chart of the solution approach is shown in Fig.~\ref{Fig:Solution_approach}.
In particular, the algorithm starts by setting the iteration index as zero and randomly choosing a feasible solution from the constraints. Then, the iteration index is increased by $1$ at the beginning of each iteration, and the feasible solution from the previous iteration is used to generate inner convex approximations of the non-convex constraints in the current iteration. Moreover, the two steps, updating the local model accuracy and solving the convex problem in \eqref{OptPrb_noneta_cvx}, are performed alternatively in an iterative manner.  

\begin{algorithm}[t]
	\caption{Proposed Algorithm for S2FL}\label{Alg:FDMA}
	\begin{algorithmic}[1]
		\State \textbf{Initialization}: Set the iteration index $\kappa = 0$ and randomly choose the initial feasible solution.
		
		\Repeat
		
    		\State Set $\kappa \leftarrow \kappa + 1$.
    		
    		\State Update the local model accuracy by solving \eqref{OptPrb_eta_cvx}.
    		
    		\State Obtain $(\boldsymbol{\tau}^{(\kappa)}, \boldsymbol{f}^{(\kappa)}, \boldsymbol{b}^{(\kappa)}, \boldsymbol{p}^{(\kappa)})$ by solving \eqref{OptPrb_noneta_cvx}. 
		
		\Until{$|O^{(\kappa)} - O^{(\kappa-1)}|/O^{(\kappa-1)} \leq \varepsilon$.}
		
		\State \textbf{Output}: the optimal solution of $(\boldsymbol{\tau}^{*}, \boldsymbol{f}^{*}, \boldsymbol{b}^{*}, \boldsymbol{p}^{*}, \eta^{*})$ and the objective value of \eqref{OptPrb_rw}.
	\end{algorithmic}
\end{algorithm}

For convergence analysis, the objective value at step $\kappa$ is given as 
\begin{align} \label{Eq:obj_value}
    O^{(\kappa)} = \frac{a}{1 - \eta^{(\kappa)}} \left( \tau^{h,(\kappa)} + \tau^{s,(\kappa)} + \tau^{l,(\kappa)} + \tau^{c,(\kappa)} \right).
\end{align}
The algorithm finally stops when the stopping criterion satisfies, i.e., the relative difference between the objective values in two consecutive iterations is less than a threshold $\varepsilon$, i.e., 
\begin{equation}
    \frac{|O^{(\kappa)} - O^{(\kappa-1)}|}{O^{(\kappa-1)}} \leq \varepsilon.
\end{equation}

The objective value of the total completion time, as defined above in \eqref{Eq:obj_value}, is bounded and reduced after each optimization iteration, i.e., the feasible solution at step $\kappa$ is better than one at step $(\kappa-1)$. Therefore, Algorithm~\ref{Alg:FDMA} finally converges to a locally optimal solution after a number of iterations. For complexity analysis, denote by $T$ the number of iterations required by the algorithm to converge, which will be shown to be relatively small via the numerical results in Section~\ref{Sec:Evaluation}. The complexity of Algorithm~\ref{Alg:FDMA} comprises two aspects: one from the update of the local model accuracy and the other from the CVX solving of the resource allocation problem. The convex problem in \eqref{OptPrb_noneta_cvx} involves $(4N+4)$ optimization variables and $6N+3$ constraints; therefore, the computational complexity of solving it is $\mathcal{O}((4N+4)^{2}(6N+3)^{2.5} + (6N+3)^{3.5})$ \cite{sheng2018power}, where $\mathcal{O}(\cdot)$ denotes the big-O notation. Moreover, the local model accuracy is updated by solving the convex problem in \eqref{OptPrb_eta_cvx}. That problem has two optimization variables (i.e., $\eta$ and $\tau^{l}$) and $(2N+2)$ constraints, thus the computational complexity of solving \eqref{OptPrb_eta_cvx} is $\mathcal{O}(4(2N+2)^{2.5} + (2N+2)^{3.5})$. Accordingly, the computational complexity of Algorithm~\ref{Alg:FDMA} is $\mathcal{O}(T((4N+4)^{2}(6N+3)^{2.5} + (6N+3)^{3.5} + 4(2N+2)^{2.5} + (2N+2)^{3.5}))$. 

\section{Extension to NOMA-enabled Transmission}
\label{Sec:Extension_NOMA}
NOMA has been considered a key technology in future wireless networks. The notion behind the NOMA technique is to allow multiple users to share the same resources; therefore, NOMA offers various advantages (e.g., higher spectral efficiency, better cell-edge throughput, and lower transmission latency) over conventional OMA schemes like FDMA. In this section, to answer the question ``can NOMA help to reduce the total completion time in the S2FL system," we consider employing NOMA for the transmission of local model updates from MDs to the edge. 

In the NOMA case, MDs can transmit the local model updates simultaneously, and the edge server decodes the received signal based on the successive interference cancellation (SIC) order. More specifically, in the uplink power-domain NOMA, signals from MDs are transmitted to the receiver (i.e., access point), where the SIC step is performed. Without loss of generality, the decoding order is the decreasing order of the channel gain qualities, as assumed in many prior studies \cite{pham2019coalitional, liang2022data}. However, we add a note that other decoding order strategies are still applicable, which is out of the research scope of this work. The channel gains are ordered as $h_{1} \geq h_{2} \geq \ldots \geq h_{N}$. Under these design principles, MD $n$ treats the signals of MDs with an index from $n+1$ to $N$ as interference. The transmission rate of MD $n$, when the NOMA technique is applied, is given as follows:
\begin{equation} \label{Eq:rate_NOMA}
    R_{n}^{\text{NOMA}} = B \log_{2}\left( 1 + \frac{h_{n}p_{n}}{n_{0}B + \sum\nolimits_{k = n+1}^{N}p_{k}h_{k}} \right).
\end{equation}
From \eqref{Eq:rate_NOMA}, all MDs are allowed to use the entire network bandwidth, while each MD $n$ in the FDMA case is allocated a fraction only, see \eqref{Eq:rate_FDMA}. As such, the bandwidth allocation is not optimized in the NOMA case.

Considering the same objective function of the total completion time, the problem of our investigated S2FL system in the NOMA-FL case is formulated as follows:
\begin{subequations}
    \label{OptPrb_NOMA}
	\begin{align} 
        & \underset{\boldsymbol{P}, \boldsymbol{p}, \boldsymbol{f}, \boldsymbol{\tau}, \eta}{\minimize}
        & & \frac{a}{1 - \eta} \left( \tau^{h} + \tau^{s} + \tau^{l} + \tau^{c} \right), \label{OptPrb_NOMA:Obj}\\
	& \text{subject to} 
        & & \eqref{OptPrb:energy}, \eqref{OptPrb:localtime}, \eqref{OptPrb:maxenergy}, \eqref{OptPrb:sensingdatasize}, \eqref{OptPrb:power}, \eqref{OptPrb:power_BS}, \eqref{OptPrb:CPUfre}, \eqref{OptPrb:accuracy} \notag\\
        &&& \tau^{c}R_{n}^{\text{NOMA}} \geq s, \forall n \in \mathcal{N}. \label{OptPrb_NOMA:datasize}
	\end{align}
\end{subequations}
It should be noted that since the bandwidth allocation vector is not optimized in this case, the bandwidth allocation notation is removed from the objective function, and the constraint \eqref{OptPrb:bAllocation} is not considered.
The problem in \eqref{OptPrb_NOMA} is highly non-convex due to the non-convex objective and non-convex constraints; therefore, solving it is computationally challenging. In the following, we develop an efficient path-following algorithm to obtain the optimal solution in the NOMA case.

Similar to the FDMA case, we leverage the decomposition technique to decompose \eqref{OptPrb_NOMA} into two smaller problems of the local model accuracy and resource allocation. The model training problem of the local model accuracy for given $(\boldsymbol{P}, \boldsymbol{p}, \boldsymbol{f}, \boldsymbol{\tau})$ is the same as one in the FDMA case. Therefore, one can solve the convex problem in \eqref{OptPrb_eta_cvx} to update the local model accuracy for a given feasible solution $(\boldsymbol{P}, \boldsymbol{p}, \boldsymbol{f}, \boldsymbol{\tau})$ of the second problem. We now consider the following optimization problem to optimize $(\boldsymbol{P}, \boldsymbol{p}, \boldsymbol{f}, \boldsymbol{\tau})$:
\begin{subequations}
    \label{OptPrb_NOMA_pftau}
	\begin{align} 
        & \underset{\boldsymbol{P}, \boldsymbol{p}, \boldsymbol{f}, \boldsymbol{\tau}}{\minimize}
        & & \eta_{\text{glo}}^{(\kappa)} \left( \tau^{h} + \tau^{s} + \tau^{l} + \tau^{c} \right), \label{OptPrb_NOMA_pftau:Obj}\\
        & \text{subject to} 
        & & \eqref{OptPrb:energy}, \eqref{OptPrb:localtime}, \eqref{OptPrb:maxenergy}, \eqref{OptPrb:sensingdatasize}, \eqref{OptPrb:power}, \eqref{OptPrb:power_BS}, \eqref{OptPrb:CPUfre}, \notag\\
        &&& \tau^{c}\log\left( 1 + \frac{h_{n}p_{n}}{n_{0}B + \sum\nolimits_{k = n+1}^{N}p_{k}h_{k}} \right) \geq \hat{s}, \notag\\
        &&& \qquad \forall n \in \mathcal{N}, \label{OptPrb_NOMA_pftau:datasize}
	\end{align}
\end{subequations}
where $\hat{s} = s B^{-1} \log2$.
Note that the feasible solution set composed of \eqref{OptPrb:maxenergy}, \eqref{OptPrb:power}, \eqref{OptPrb:power_BS}, and \eqref{OptPrb:CPUfre} is convex, while we leverage the same approaches in the FDMA case to generate innerly convex approximations of the non-convex constraints in \eqref{OptPrb:energy} and \eqref{OptPrb:localtime}. In the following, we process the non-convex constraint in~\eqref{OptPrb_NOMA_pftau:datasize} through an inner convex approximation in order to generate a convex program of the non-convex problem in~\eqref{OptPrb_NOMA_pftau}.

By applying the inequality in~\eqref{Ineq:logxyz} for the left-hand side of \eqref{OptPrb_NOMA_pftau:datasize} with $z = \tau^{c}$, $x = 1/p_{n}h_{n}$, $y = n_{0}B + \sum\nolimits_{k = n+1}^{N}p_{k}h_{k}$, $\bar{z} = \tau^{c,{(\kappa-1)}}$, $\bar{x} = 1/p_{n}^{(\kappa-1)}h_{n}$, and $\bar{y} = n_{0}B + \sum\nolimits_{k = n+1}^{N}p_{k}^{(\kappa-1)}h_{k}$, we can innerly approximate the left-hand side of \eqref{OptPrb_NOMA_pftau:datasize} as follows: 
\begin{align} \label{Ineq:Rate_NOMA}
& \tau^{c}\log\left( 1 + \frac{h_{n}p_{n}}{n_{0}B + \sum\nolimits_{k = n+1}^{N}p_{k}h_{k}} \right) \geq  \notag\\ 
& \; \lambda_{k} + \mu_{k}\left( 2 - \frac{p_{n}^{(\kappa-1)}}{p_{n}} - \frac{n_{0}B + \sum\nolimits_{k = n+1}^{N}p_{k}h_{k}}{n_{0}B + \sum\nolimits_{k = n+1}^{N}p_{k}^{(\kappa-1)}h_{k}} \right) - \frac{\upsilon_{k}}{\tau^{c}} \notag\\
& \; \triangleq S_{n}^{(\kappa)}(p_{n},\tau^{c}),
\end{align}
where the coefficients are given as follows: 
\begin{align}
    & \lambda_{k} = \tau^{c,{(\kappa-1)}}\log\left( 1 + \frac{h_{n}p_{n}^{(\kappa-1)}}{n_{0}B + \sum\nolimits_{k = n+1}^{N}p_{k}^{(\kappa-1)}h_{k}} \right), \\
    & \mu_{k} = \tau^{c,{(\kappa-1)}}\left( 1 + \frac{n_{0}B + \sum\nolimits_{k = n+1}^{N}p_{k}^{(\kappa-1)}h_{k}}{h_{n}p_{n}^{(\kappa-1)}} \right)^{-1}, \\
    & \upsilon_{k} = (\tau^{c,{(\kappa-1)}})^{2}\log\left( 1 + \frac{h_{n}p_{n}^{(\kappa-1)}}{n_{0}B + \sum\nolimits_{k = n+1}^{N}p_{k}^{(\kappa-1)}h_{k}} \right).
\end{align}
Accordingly, the left-hand side of \eqref{OptPrb_NOMA_pftau:datasize} can be convexified as $S_{n}^{(\kappa)}(p_{n},\tau^{c})$ and the constraint is innerly convex approximated as follows:
\begin{equation}
    S_{n}^{(\kappa)}(p_{n},\tau^{c}) \geq \hat{s}, \forall n \in \mathcal{N}.
\end{equation}

In summary, we solve the following convex problem at step $\kappa$ to obtain the feasible solution $(\boldsymbol{P}^{(\kappa)}, \boldsymbol{p}^{(\kappa)}, \boldsymbol{f}^{(\kappa)}, \boldsymbol{\tau}^{(\kappa)})$:
\begin{subequations}
    \label{OptPrb_NOMA_pftau_cvx}
	\begin{align} 
        & \underset{\boldsymbol{P}, \boldsymbol{p}, \boldsymbol{f}, \boldsymbol{\tau}}{\minimize}
        & & U^{(\kappa)}(\boldsymbol{P}, \boldsymbol{\tau}), \label{OptPrb_NOMA_pftau_cvx:Obj}\\
	& \text{subject to} 
        & & S_{n}^{(\kappa)}(p_{n},\tau^{c}) \geq \hat{s}, \forall n \in \mathcal{N}, \label{OptPrb_NOMA_pftau_cvx:datasize}\\ 
        &&& \eqref{OptPrb_noneta_cvx:energy}, \eqref{OptPrb_noneta_cvx:localtime}, \eqref{OptPrb:maxenergy}, \eqref{OptPrb:sensingdatasize}, \eqref{OptPrb:power}, \eqref{OptPrb:power_BS}, \eqref{OptPrb:CPUfre}. \notag
	\end{align}
\end{subequations}
This problem is convex; therefore, it is readily solvable by the CVX package \cite{Grant2013CVX}.

\begin{algorithm}[t]
	\caption{Proposed Algorithm for NOMA-S2FL}\label{Alg:NOMA}
	\begin{algorithmic}[1]
		\State \textbf{Initialization}: Set the iteration index $\kappa = 0$ and randomly choose the initial feasible solution.
		
		\Repeat
		\State Set $\kappa \leftarrow \kappa + 1$.
		
		\State Solve \eqref{OptPrb_eta_cvx} to update the local model accuracy $\eta^{(\kappa)}$.
		
		\State Solve \eqref{OptPrb_NOMA_pftau_cvx} to obtain $(\boldsymbol{P}^{(\kappa)}, \boldsymbol{p}^{(\kappa)}, \boldsymbol{f}^{(\kappa)}, \boldsymbol{\tau}^{(\kappa)})$. 
		
		\Until{$|U^{(\kappa)} - U^{(\kappa-1)}|/U^{(\kappa-1)} \leq \varepsilon$.}
		
		\State \textbf{Output}: the optimal solution  $(\boldsymbol{P}^{*}, \boldsymbol{p}^{*}, \boldsymbol{f}^{*}, \boldsymbol{\tau}^{*}, \eta^{*})$, and the optimal cost value $O^{*}(\boldsymbol{P}^{*}, \boldsymbol{p}^{*}, \boldsymbol{f}^{*}, \boldsymbol{\tau}^{*})$.
	\end{algorithmic}
\end{algorithm}

The pseudo cost of the proposed procedure for solving the S2FL problem in the NOMA case is presented in Algorithm~\ref{Alg:NOMA}. Similar to Algorithm~\ref{Alg:FDMA}, Algorithm~\ref{Alg:NOMA} updates the two variable blocks in the fashion of the alternating optimization approach until the stopping criterion is satisfied. Also similar to the convergence analysis for Algorithm~\ref{Alg:FDMA}, the algorithm in the NOMA case is finally converged to a locally optimal solution. For computational analysis, the computational complexity of Algorithm~\ref{Alg:NOMA} is $\mathcal{O}(T((3N+4)^{2}(6N+2)^{2.5} + (6N+2)^{3.5} + 4(2N+2)^{2.5} + (2N+2)^{3.5}))$. It is since the convex problem in \eqref{OptPrb_NOMA_pftau_cvx} involves $2N + 4$ variables and $6N+2$ constraints, while the computational complexity of solving the model training problem is $\mathcal{O}(4(2N+2)^{2.5} + (2N+2)^{3.5})$.

\section{Performance Evaluation}
\label{Sec:Evaluation}
Simulation results are presented in this section to evaluate the proposed algorithm's performance. We first provide the simulation settings and then discuss simulation results and comparisons with the benchmark schemes. 

\subsection{Simulation Settings}
We simulate an S2FL network with $10$ single-antenna MDs for performance evaluation. The edge server collocated with an access point is equipped with $40$ antennas, which are divided equally among 10 MDs, i.e., $N_{a} = 4$ antennas are used to create energy beamforming for each MD. The average fading power between MD $n$ and edge server is modeled as $ \bar{g}_{n} = 5\times 10^{-4} \times (\varrho_{n})^{-2}$, where $\varrho_{n}$ is the distance between the edge server and MD $n$ and $\varrho_{n}$ is randomly distributed in the range $[1, 5]$ m \cite{li2019wirelessly}. The channel vector for MD $n$ is modeled by the Rician fading model as follows: 
\begin{equation}
    \boldsymbol{g}_{n} = \sqrt{\frac{K\bar{g}_{n}}{K+1}}\boldsymbol{g}_{n}^{\text{LoS}} + \sqrt{\frac{\bar{g}_{n}}{K+1}}\boldsymbol{g}_{n}^{\text{NLoS}},
\end{equation}
where $K$ denotes the Rician factor, and $\boldsymbol{g}_{n}^{\text{LoS}}$ and $\boldsymbol{g}_{n}^{\text{NLoS}}$ are $4 \times 1$ vectors representing the line-of-sight (LoS) component and non-LoS (NLoS) component, respectively. Here, $\boldsymbol{g}_{n}^{\text{LoS}}$ is a deterministic component with $|\boldsymbol{g}_{n}^{\text{LoS}}| = 1$, and $\boldsymbol{g}_{n}^{\text{NLoS}}$ follows a normal distribution with zero mean and unit variance. Then, the channel gain $h_{n}$ from MD $n$ to the edge server is computed as $h_{n} = |\boldsymbol{g}_{n}|^{2}$.
We set the transmit power budget of the energy source (i.e., edge server) to be $P_{0} = 42$ dBm, the maximum transmit power of MDs to be $p_{n}^{\max} = 10$ dBm, the system bandwidth to be $B = 500$ kHz, the noise power density to be $n_{0} = 10^{-14}$ W/Hz, and the maximum and minimum CPU frequency to be $f_{n}^{\max} = 2.0$ GHz and $f_{n}^{\min} = 0.1$ GHz, respectively \cite{pham2022energy}. We set the energy conversion efficiency coefficient $\varphi = 0.9$ and the reward coefficient $q_{n}^{r} \in [1, 10]\times 10^{-12}$ J/bit for all MDs. For the local computing model, the computation workload is randomly chosen from the range $[10, 20]$ cycles per bit, the coefficient of the CPU is $\xi = 10^{-28}$, and the model size $s = 28.1$ Kb is identical for all MDs. For the data sensing model, the sensing rate $r_{n}$ is picked at random from the range $[1, 5]\times 10^{6}$ bit/s, the energy consumption per bit $q_{n}^{s}$ is uniformly distributed in the range $[1, 10]\times 10^{-12}$ J/bit \cite{li2019wirelessly}, and the required sensing data size is $D_{0} = 100$ Kb. 

\subsection{Performance Results}
\begin{figure}[t]
	\centering
	\includegraphics[width=0.885\linewidth]{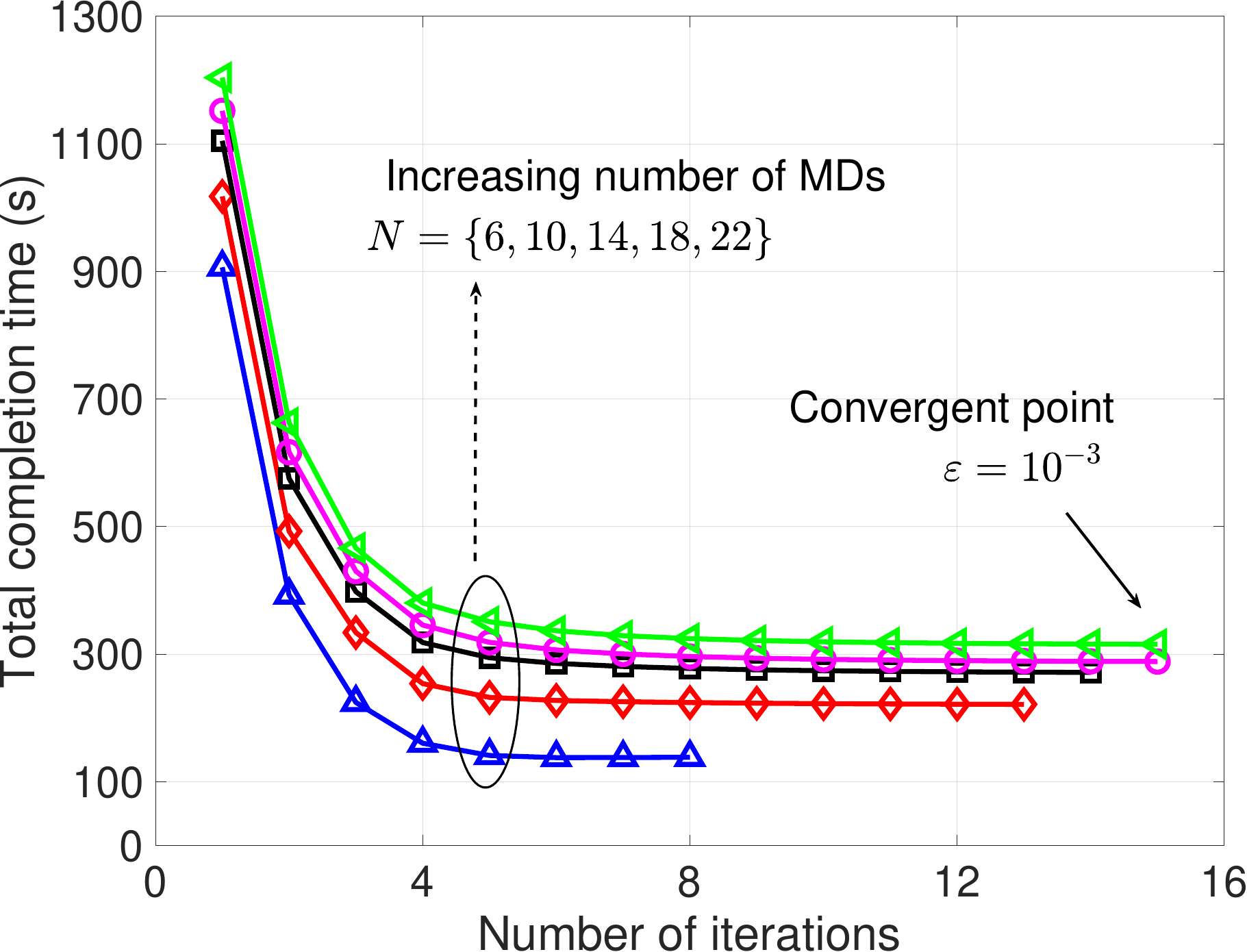}
	\caption{Convergence curve of the proposed algorithm.}
	\label{Fig:Convergence_vs_users}
\end{figure}

Fig.~\ref{Fig:Convergence_vs_users} shows the convergence curve of the proposed algorithm for a given network realization with different numbers of MDs, $N = \{6, 10, 14, 18, 22\}$, where the stopping threshold is set to be $\varepsilon = 10^{-3}$, and each energy beam is served by $4$ antennas at the source. It is observed from the figure that the proposed algorithm only requires around $8$-$15$ iterations to be converged, thus showing the computational effectiveness of Alg.~\ref{Alg:FDMA}. The figure also shows that the algorithm with higher MDs typically takes more iterations to converge to the optimal solution. This is because the feasible solution set and the computational complexity become larger with the increment in the number of MDs $N$. Furthermore, the total completion time increases with the number of MDs. The reason is that the higher $N$ is, the higher the possibility that MDs have worse effective channel gains and harvest less energy. 

In order to analyze the performance improvements of our joint power transfer, data sensing, model training, and transmission framework and the proposed algorithm (labeled as S2FL), a set of benchmark schemes is considered. Firstly, in the \textit{fixed training data} (FTD) scheme (corresponding to \cite{pham2022energy}), which jointly optimizes power transfer, model training, and data transmission for given training data. Thus, the sensing time is configured with a preset value subject to the training data constraint in \eqref{Eq:trainingdata}, and there is no need to consider optimizing $\tau^{s}$. Secondly, in the \textit{fixed local model accuracy} (FLA) scheme (corresponding to \cite{li2022joint}), our proposed S2FL algorithm is performed under given data processing at MDs ($\eta = 0.25$). Accordingly, one can pre-calculate the number of local rounds $\eta_{\text{loc}}$ via \eqref{Eq:localrounds} and global communication rounds $\eta_{\text{glo}}$ via \eqref{Eq:globalrounds} in the FLA scheme. Accordingly, the corresponding algorithm in the GLA scheme only comprises the path-following procedure for solving the resource allocation problem. Thirdly, we compare the proposed algorithm with the \textit{proportional power transfer} (PPT) scheme (corresponding to \cite{wu2022non}) under fixed transmit power for WPT. For a fair allocation among MDs, the power budget of the access point is allocated to $N$ beams according to the effective channel gains between MDs and the energy source. Finally, the proposed algorithm is compared with the \textit{equal bandwidth allocation} (EBA) scheme, where the system bandwidth is equally allocated to MDs. Moreover, the solutions for these benchmark schemes are readily obtained after some minor modifications of Algorithm~\ref{Alg:FDMA}. Note that the following results are obtained by taking an average of multiple network realizations, each randomizes the locations of MDs.   

\begin{figure}[t]
	\centering
	\includegraphics[width=0.885\linewidth]{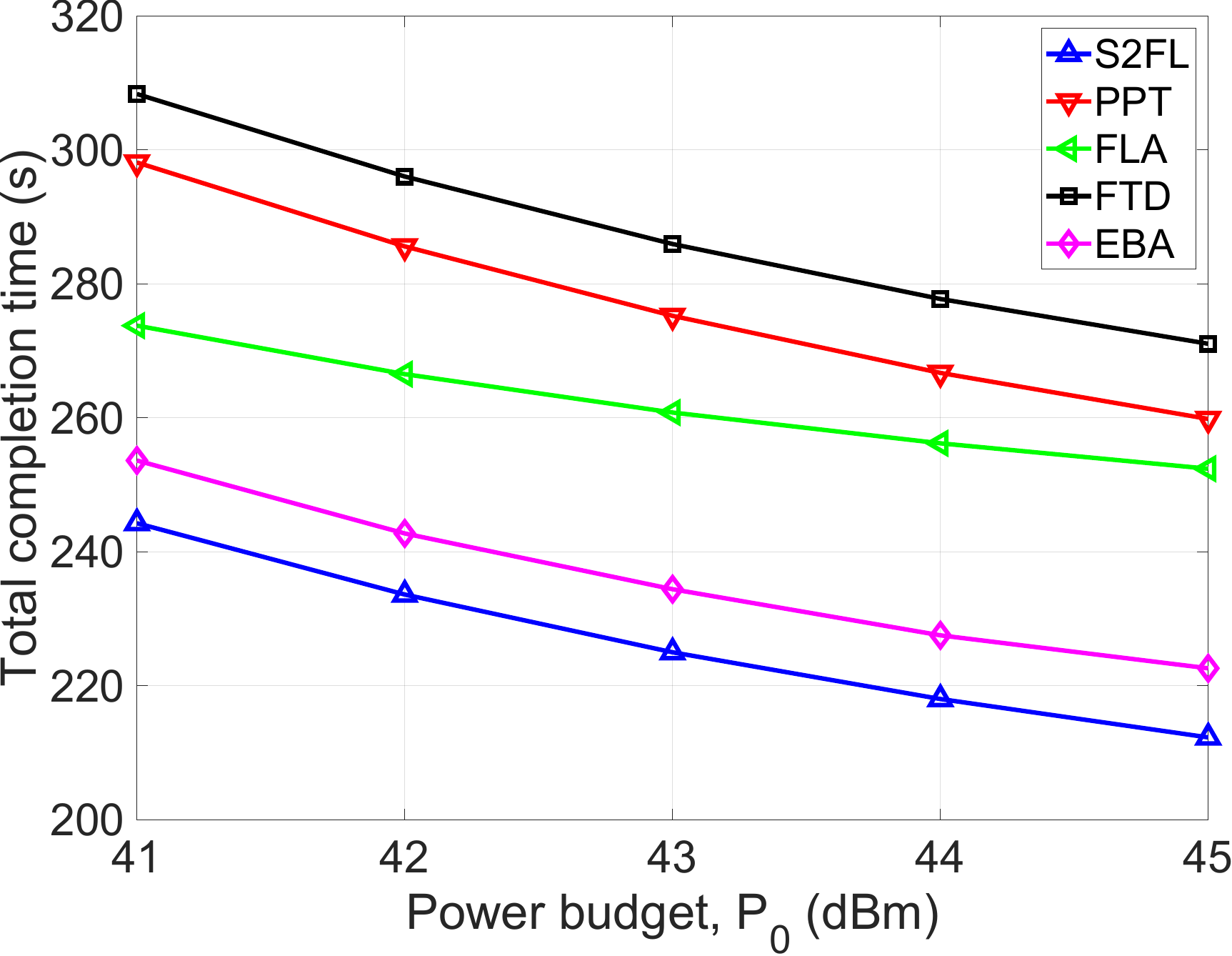}
	\caption{Total completion versus the power budget $P_{0}$ of the energy source.}
	\label{Fig:TCT_vs_powerbudget}
\end{figure}

\begin{figure}[t]
	\centering
	\includegraphics[width=0.885\linewidth]{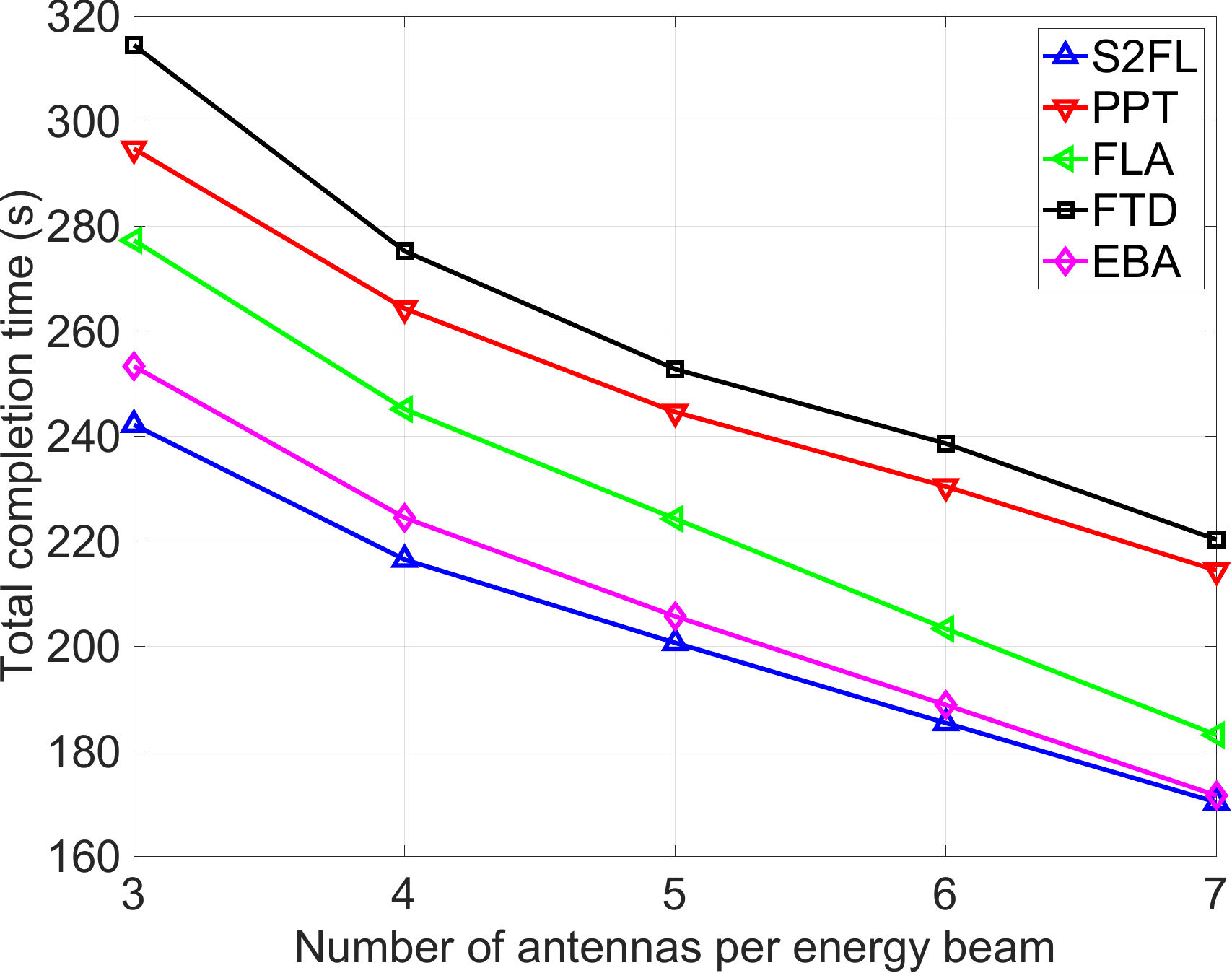}
	\caption{Total completion versus the antenna setting at the energy source.}
	\label{Fig:TCT_vs_antennasperbeam}
\end{figure}

We show the performance of the proposed algorithm with different power budgets of the energy source. It is shown in Fig.~\ref{Fig:TCT_vs_powerbudget} that the total completion time gets smaller when the power budget $P_{0}$ of the energy source becomes larger. It is because increasing $P_{0}$ helps MDs to harvest more energy for learning purposes. Therefore, with larger $P_{0}$, one can reduce the time durations for the learning phases and consequently shorten the total completion time. Notably, the proposed algorithm achieves the lowest completion time over the benchmark schemes, which respectively preset the power transfer, local model accuracy, training data, or bandwidth assignment. Accordingly, Fig.~\ref{Fig:TCT_vs_powerbudget} also confirms the superiority of the S2FL framework and proposed algorithm in terms of total completion time. On average, the proposed algorithm offers a time reduction of around $21.24$\%, $18.21$\%, $13.42$\%, and $4.03$\% compared with the FTD, PPT, FLA, and EBA schemes, respectively.
We also evaluate the total completion time of different schemes when the number of antennas per energy beam $N_{a}$ is varied from $3$ to $7$. Similar to the observation from Fig.~\ref{Fig:TCT_vs_powerbudget}, one can observe from Fig.~\ref{Fig:TCT_vs_antennasperbeam} that the total completion time decreases with the number of antennas at the energy source. This is due to the fact that larger $N_{a}$ helps the energy source better transfer power to MDs for the same power allocation at the energy source, and thus MDs can harvest more energy for learning purposes, which in turn reduces the completion time. 

\begin{figure}[t]
	\centering
	\includegraphics[width=0.885\linewidth]{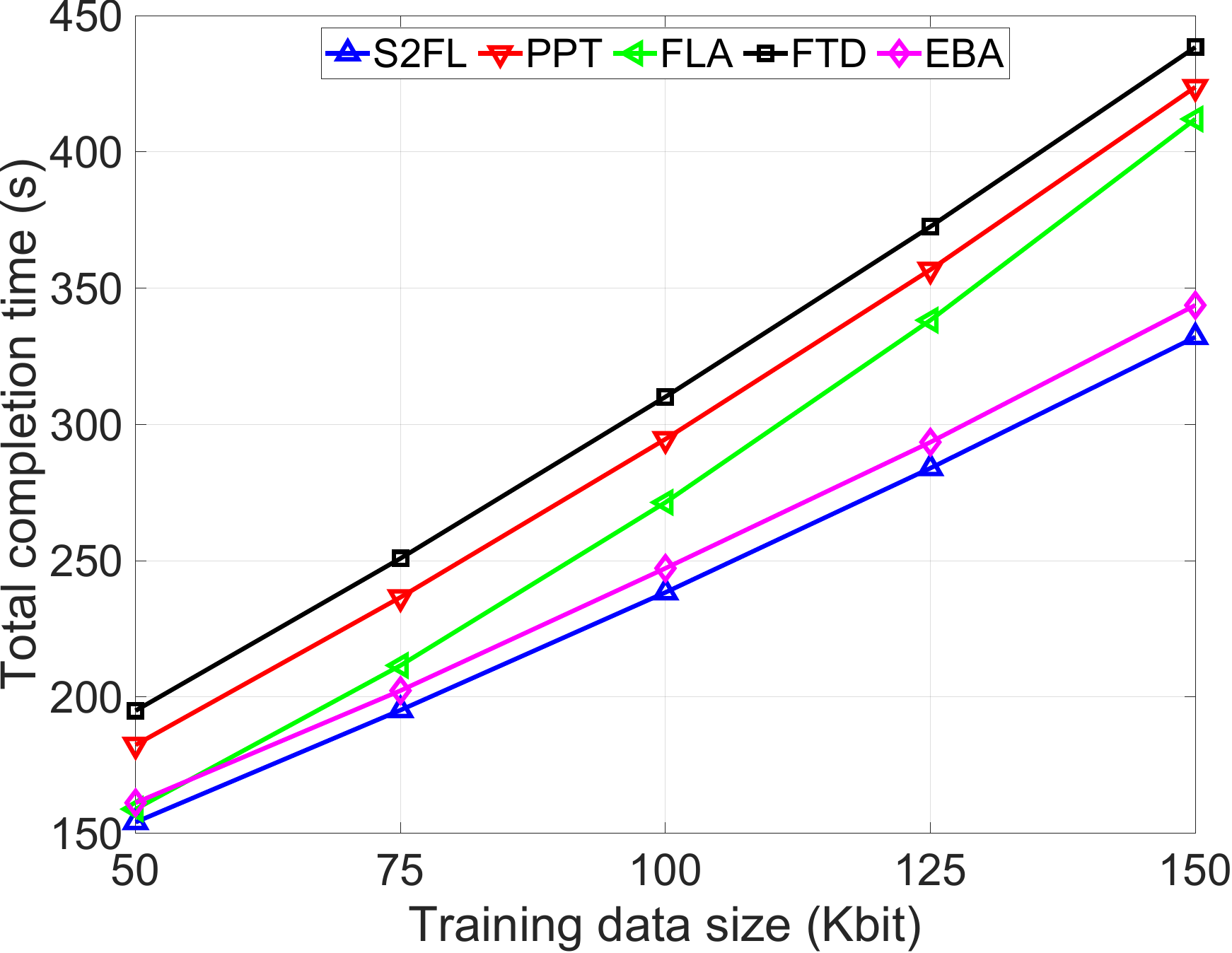}
	\caption{Total completion versus the amount of sensing data $D_{0}$ in each round.}
	\label{Fig:TCT_vs_trainingdatasize}
\end{figure}

We plot the total completion time when varying the amount of data $D_{0}$ needed to be sensed in each learning round. As shown in Fig.~\ref{Fig:TCT_vs_trainingdatasize}, upon increasing $D_{0}$, the total completion time achieved by all the algorithms increases. This can be explained since the larger $D_{0}$ is, the higher the sensing time required by MDs is. Accordingly, for the same harvested energy, MDs need to increase the time duration of the energy harvesting time, or one needs to decrease the local training's accuracy (i.e., $\eta$ is higher and $\eta_{\text{loc}}$ is smaller), as mathematically illustrated in the energy constraint~\eqref{OptPrb_rw:energy}. As such, both the completion time per round and the number of global communication rounds increase, thus leading to a higher total completion time. Again, Fig.~\ref{Fig:TCT_vs_trainingdatasize} shows that the proposed S2FL algorithm is the best performer.

\begin{figure}[t]
	\centering
	\includegraphics[width=0.885\linewidth]{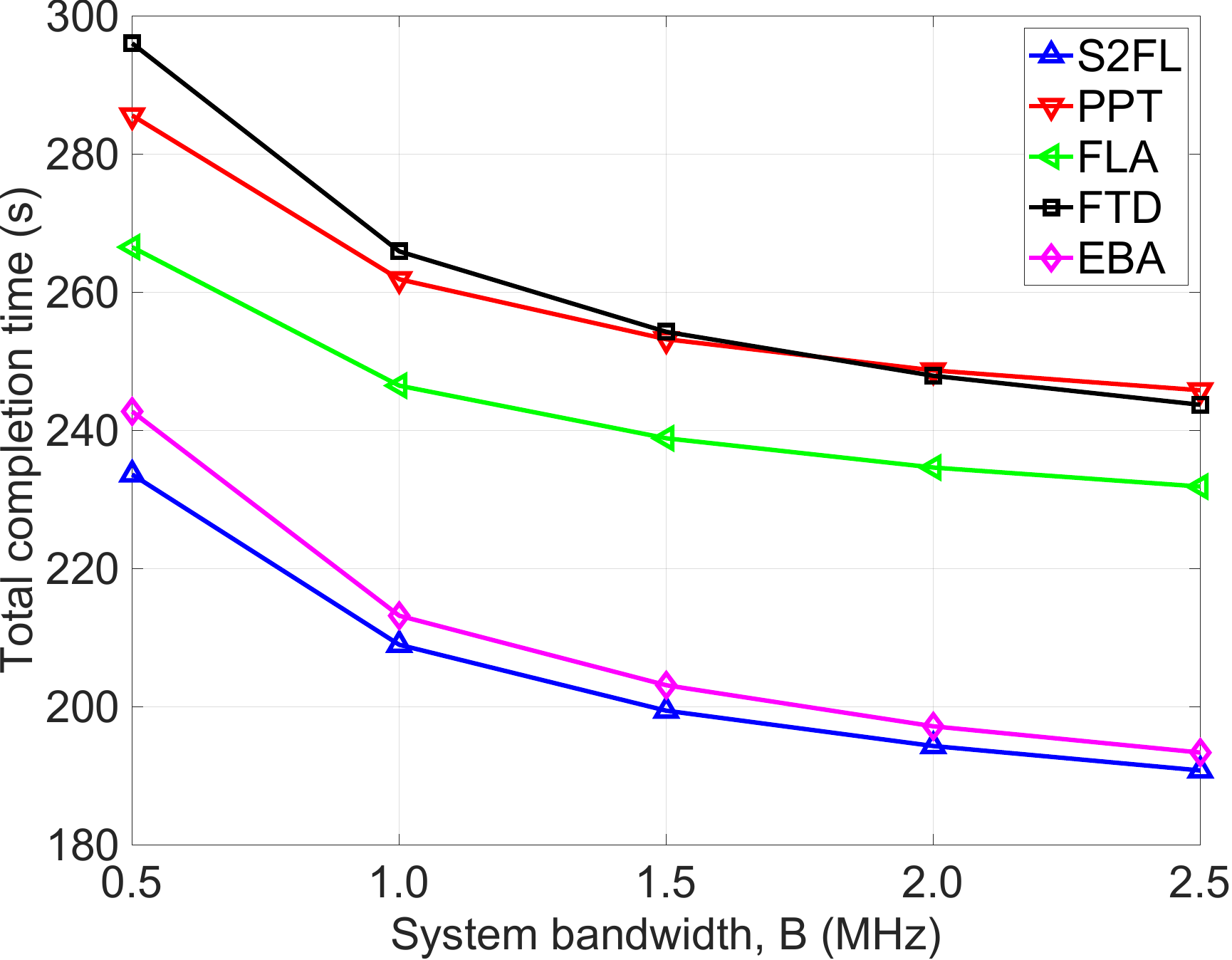}
	\caption{Total completion versus the system bandwidth $B$.}
	\label{Fig:TCT_vs_systembandwidth}
\end{figure}

We show the performance comparison when the system bandwidth varies from $0.5$ MHz to $2.5$ MHz with a deviation step of $0.5$ MHz. It is shown in Fig.~\ref{Fig:TCT_vs_systembandwidth} that the larger the system bandwidth is, the lower the total completion time of the FL system is. The reason is that more available bandwidth resources can increase the transmission rate between the edge server and MDs, thus reducing the transmission time and the total completion time as well. 
In comparison to the FTD, FLA, and PPT benchmarks, the S2FL algorithm achieves a respective average of $21.45$\%, $15.60$\%, and $20.61$\% performance improvement in terms of total completion time. Furthermore, the EBA scheme linearly approaches the performance of our proposed S2FL algorithm. For example, compared to the EBA scheme, S2FL provides a performance gain of around $3.76$\% and $1.34$\% when $B = 0.5$ MHz and $B = 2.5$ MHz, respectively. This is because more bandwidth resources relax the need for bandwidth allocation in S2FL. 

\begin{figure}[t]
	\centering
	\includegraphics[width=0.885\linewidth]{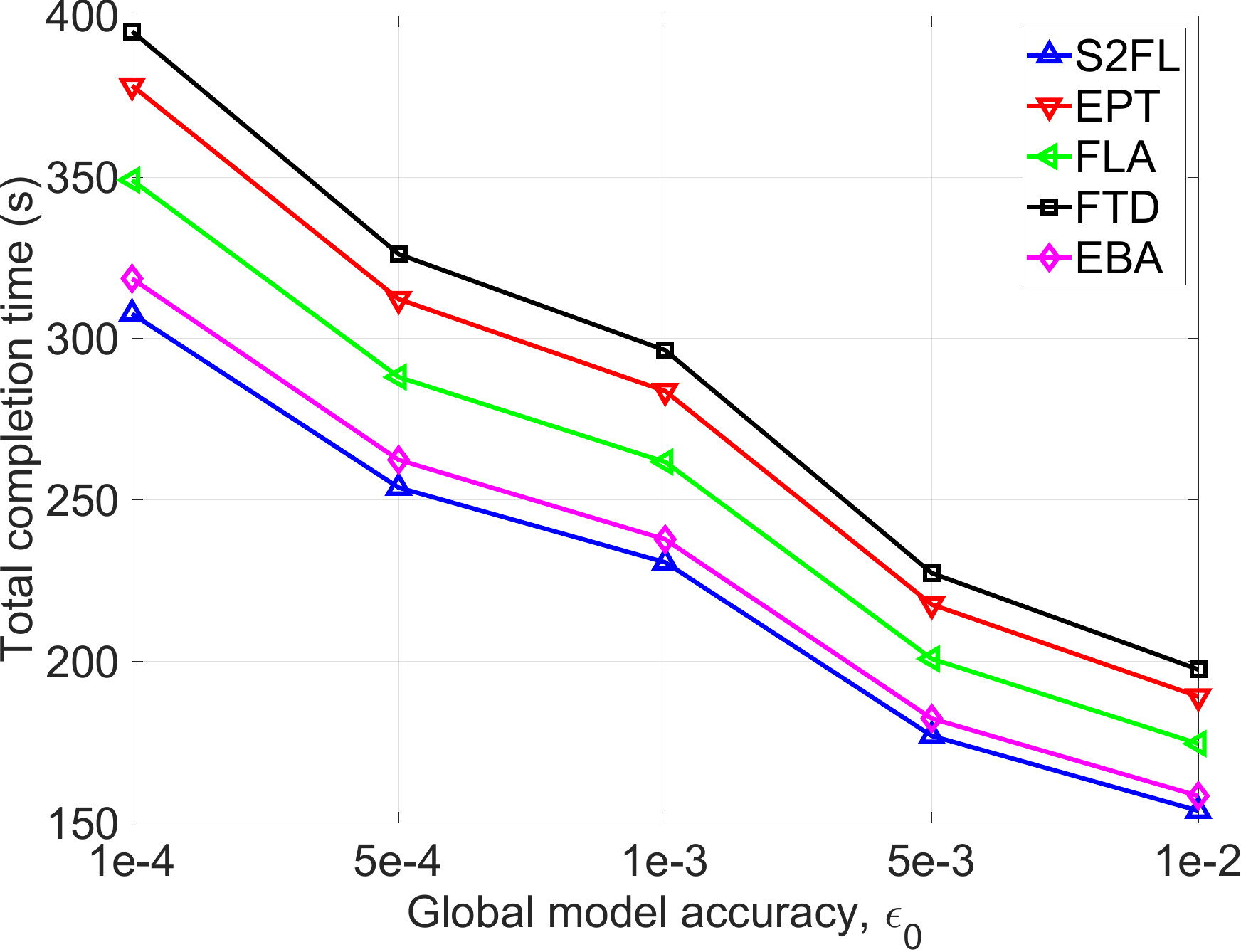}
	\caption{Total completion versus the global model accuracy $\epsilon_{0}$.}
	\label{Fig:TCT_vs_epsilon0}
\end{figure}

We plot the total completion time in Fig.~\ref{Fig:TCT_vs_epsilon0} when varying the accuracy requirement for the global model, $\epsilon_{0}$. One can observe from the figure that the completion time reduces with $\epsilon_{0}$, i.e., the FL system has a higher completion time when the desired accuracy of the global model is higher. This is because the required number of global communication rounds $\eta_{\text{glo}}$ gets larger when $\epsilon_{0}$ is smaller, as mathematically modeled in \eqref{Eq:globalrounds}. One can also observe that the performance gap between the proposed S2FL algorithm and the benchmark schemes is smaller with $\epsilon_{0}$, which coincides precisely with the observation in \cite{pham2022energy}, where the system energy consumption is minimized. However, it would be noted that the work in \cite{pham2022energy} assumes the availability of training data as in the FTD scheme, showing that our proposed framework of joint power transfer, data sensing, model training, and transmission can bring performance improvements over the alternative schemes with only a set of optimized parameters.

\subsection{Performance Comparison in S2FL: OMA vs NOMA}

\begin{figure}[t]
	\centering
	\subfloat[Total completion time vs. the system bandwidth $B$\label{Fig:TCT_vs_B_OMA_NOMA}]{\includegraphics[width=0.885\linewidth]{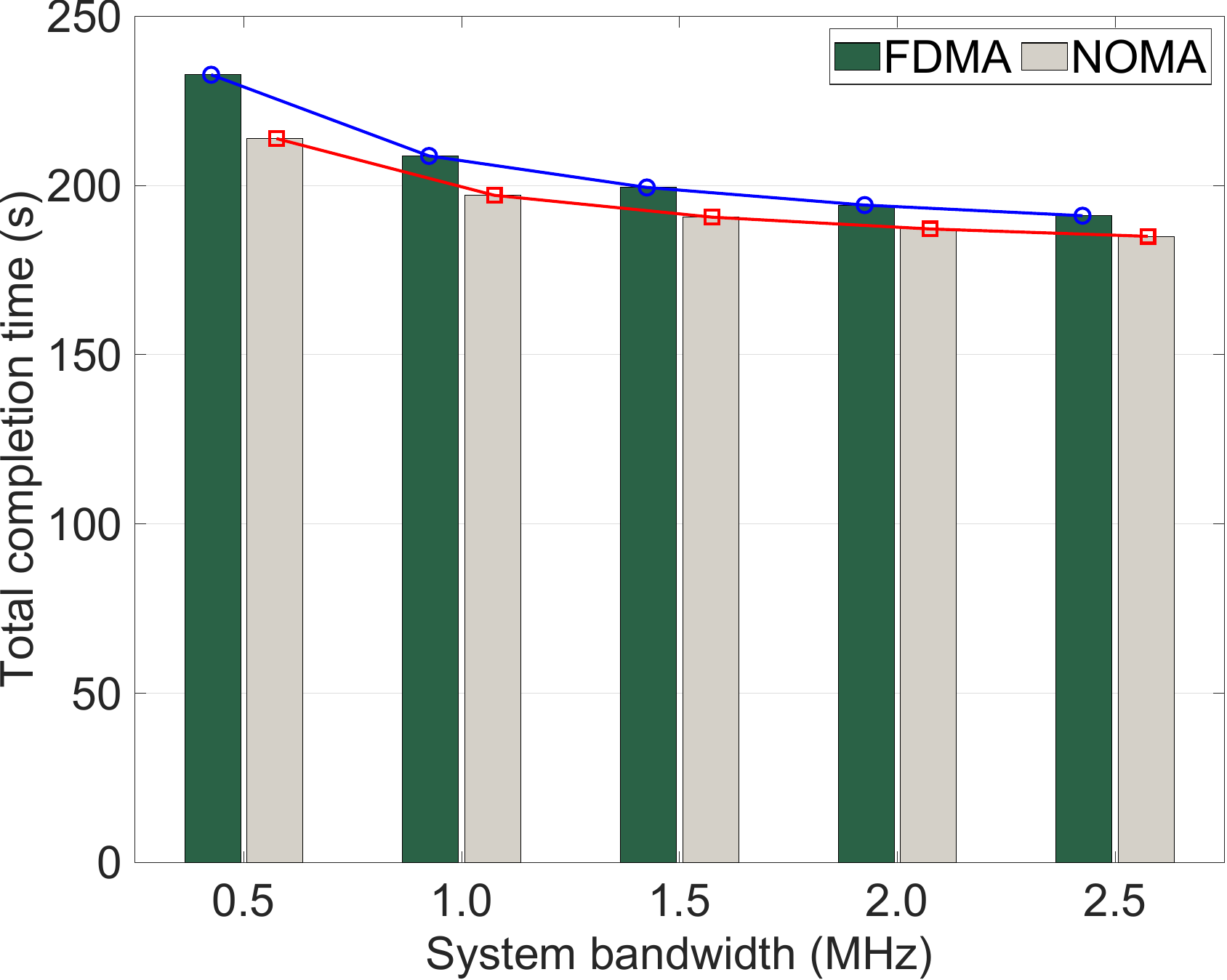}}\qquad
	\subfloat[Total completion time vs. the power budget $P_{0}$\label{Fig:TCT_vs_P0_OMA_NOMA}]{\includegraphics[width=0.885\linewidth]{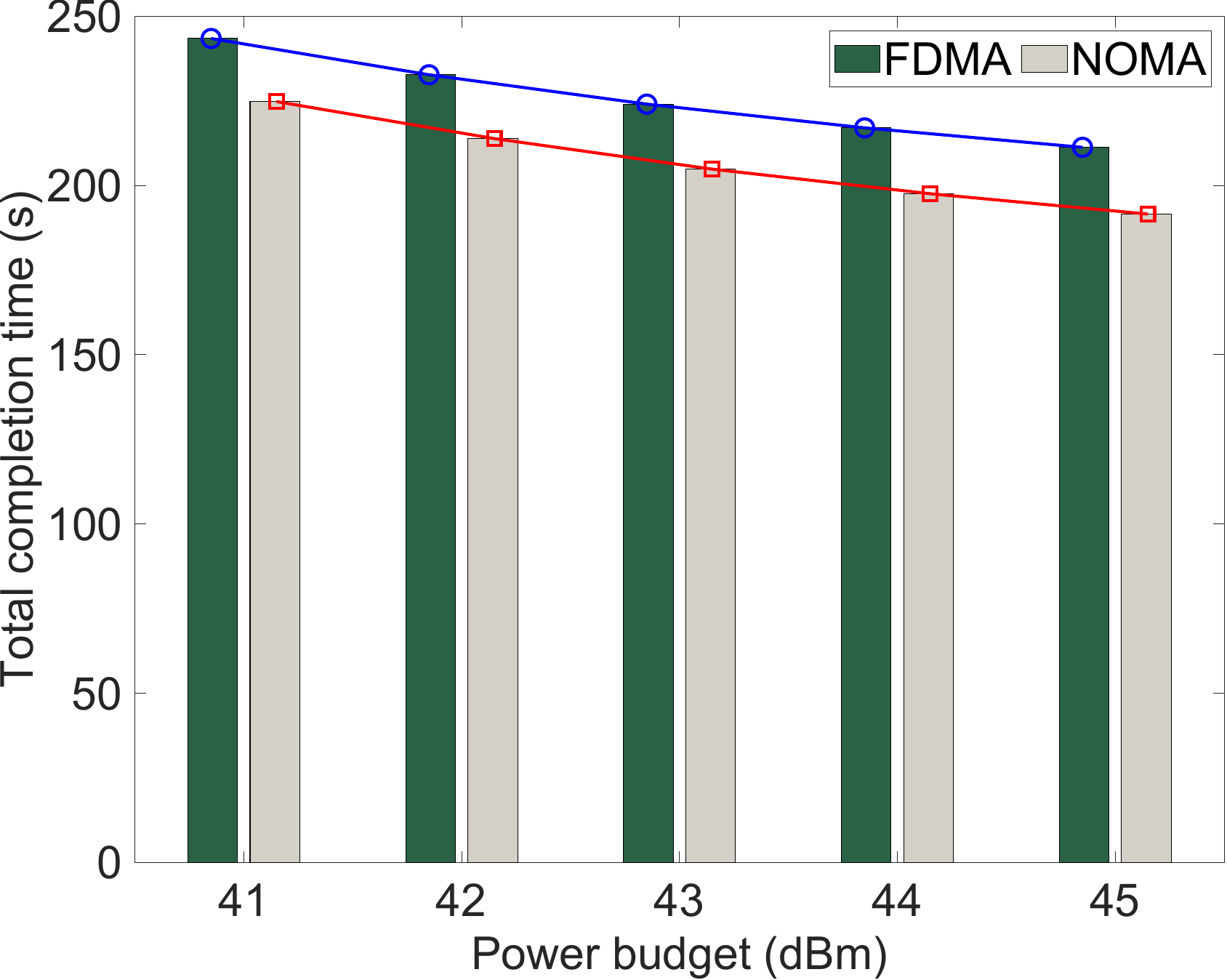}}
	\caption{Performance comparison between OMA and NOMA.}
	\label{Fig:PerformanceComparison}
\end{figure} 

In the previous part, we have shown that the proposed S2FL algorithm is able to bring significant improvements in terms of completion time when it is compared with several benchmark schemes. Now, we compare the S2FL system in the FDMA (i.e., OMA) case with its NOMA counterpart. Fig.~\ref{Fig:PerformanceComparison} clearly indicates that the completion time is considerably reduced in the NOMA case. For example, with the scenarios of $N = 10$ MDs, the NOMA-enabled S2FL can reduce approximately $18.8618$ s and $6.0910$ s when the system bandwidth is $0.5$ MHz and $2.5$ MHz, respectively. Fig.~\ref{Fig:TCT_vs_P0_OMA_NOMA} also indicates that more power budget resources at the edge server help enlarge the performance gain achieved by NOMA and OMA, i.e., $8.36$\% on average. Note that this work aims to develop a synchronous design of the S2FL system while FDMA and NOMA are employed for the transmission phase of all MDs. Therefore, the results in Fig.~\ref{Fig:PerformanceComparison} evidently demonstrate the performance superiority of NOMA over FDMA in data transmission when both are equipped with the optimal resource solutions in the considered synchronous S2FL network. It is promising to consider extension scenarios when MDs may have different time durations of the sensing, training, and transmission phases (i.e., asynchronous designs). As such, FMDA, with the ability to ensure asynchronicity for MDs, may also outweigh the performance of its synchronous counterpart and the NOMA-enabled scheme. We leave this exciting research for future work.

\section{Conclusion}
\label{Sec:Conclusion}
This work has investigated the S2FL framework by leveraging MDs' data sensing and energy harvesting capabilities. Notably, the proposed S2FL framework enables energy-limited MDs to wirelessly harvest energy from the energy source and collect the sensing data from the environment for local training. Following the S2FL framework, a computationally-efficient path-following algorithm has been proposed to solve the computationally challenging non-convex problem in S2FL. We have shown that the proposed algorithm provides a significant reduction in total completion time compared with a set of benchmark schemes under various network settings. We have also shown that adopting NOMA as the multiple access technique helps further reduce the total completion time of the considered S2FL network. 



\end{document}